\documentclass[letterpaper,11pt]{article}

\usepackage{amsmath,amsthm,amsfonts,latexsym,graphicx,amssymb}
\usepackage[margin=1in]{geometry}
\usepackage{xcolor}
\usepackage{url,hyperref}
\usepackage{paralist}
\usepackage{cases}
\usepackage[noend]{algorithmic}
\usepackage[boxed]{algorithm}
\usepackage{paralist}
\usepackage{subfiles}
\usepackage{import}
\usepackage{stmaryrd}

\usepackage{enumerate}

\usepackage{subfigure}

\newtheorem{theorem}{Theorem}
\newtheorem{lemma}[theorem]{Lemma}
\newtheorem{claim}[theorem]{Claim}
\newtheorem{prop}[theorem]{Proposition}

\newtheorem{definition}[theorem]{Definition}
\newtheorem{observation}[theorem]{Observation}

\newcommand{\poly}{\mathtt{poly}}

\newcommand{\cF}{\ensuremath{\mathcal{F}}}

\newcommand{\cA}{\ensuremath{\mathcal{A}}}
\newcommand{\cK}{\ensuremath{\mathcal{K}}}
\newcommand{\cc}{\ensuremath{\mathrm{cc}}}

\newcommand{\tw}{\ensuremath{\mathrm{tw}}}

\newcommand{\best}{\ensuremath{\mathsf{Best}}}
\newcommand{\Ints}{\mathbb{Z}}
\renewcommand{\to}{\rightarrow}
\newcommand{\lb}{\llbracket}
\newcommand{\rb}{\rrbracket}

\newcommand{\ohs}{\ensuremath{O^*}}

\DeclareMathOperator{\dist}{dist}

\newcommand{\pGT}{{\sc Grid Tiling}}
\newcommand{\pST}{{\sc Steiner Tree}}
\newcommand{\pPST}{{\sc Planar Steiner Tree}}
\newcommand{\pSAT}{{\sc 3-SAT}}
\newcommand{\pPBSF}{{\sc Planar Block Steiner Forest}}

\newcommand{\algcomment}[1]{\colorbox{black!10}{#1}}

\newcommand{\knip}[1]{#1}
\newlength{\problemoffset}
\newcommand{\optimization}[3]{
		\begin{list}{}{
				\setlength{\leftmargin}{\problemoffset}
				\setlength{\rightmargin}{\problemoffset}
				\setlength{\parsep}{0pt}
				\setlength{\itemsep}{2pt}
				\setlength{\topsep}{\itemsep}
				\setlength{\partopsep}{\itemsep}
			}
			\item
			{\textbf{#1}}
			\item
			{\textbf{Instance}: #2}
			\item
			{\textbf{Asked}: #3}
		\end{list}
}

\title{Nearly ETH-Tight Algorithms for {\pPST} with Terminals on Few Faces}
\author{
	S\'andor Kisfaludi-Bak\thanks{Eindhoven University of Technology, The Netherlands; \url{s.kisfaludi.bak@tue.nl}. Supported by the Netherlands Organization for Scientific Research NWO under project no. 024.002.003.}\quad
	Jesper Nederlof\thanks{Eindhoven University of Technology, The Netherlands; \url{j.nederlof@tue.nl}. Supported by the Netherlands Organization for Scientific Research NWO under project no. 024.002.003 and no. 639.021.438.}\quad
	Erik Jan van Leeuwen\thanks{Utrecht University, The Netherlands; \url{e.j.vanleeuwen@uu.nl}.}
}

\begin{document}
	\maketitle
\thispagestyle{empty}
\begin{abstract}
	The {\pST} problem is one of the most fundamental NP-complete problems as it models many network design problems.
	Recall that an instance of this problem consists of a graph with edge weights, and a subset of vertices (often called terminals); the goal is to find a subtree of the graph of minimum total weight that connects all terminals.
	A seminal paper by Erickson et al. [Math. Oper. Res., 1987] considers instances where the underlying graph is planar and all terminals can be covered by the boundary of $k$ faces. Erickson et al. show that the problem can be solved by an algorithm using $n^{O(k)}$ time and $n^{O(k)}$ space, where $n$ denotes the number of vertices of the input graph. In the past 30 years there has been no significant improvement of this algorithm, despite several efforts.
	
	In this work, we give an algorithm for {\pPST} with running time $2^{O(k)} n^{O(\sqrt{k})}$ using only polynomial space. Furthermore, we show that the running time of our algorithm is almost tight: we prove that there is no $f(k)n^{o(\sqrt{k})}$ algorithm for {\pPST} for any computable function $f$, unless the Exponential Time Hypothesis fails.
\end{abstract}
\clearpage
\setcounter{page}{1}
\section{Introduction}
\newcommand{\maxweight}{W}
In the {\pST} problem, we are given an undirected $n$-vertex graph $G$ with edge weights\footnote{For convenience, we assume all comparisons and additions of weights take $O(1)$ time.} $\omega: E(G) \rightarrow \{0,\ldots,\maxweight\}$ and a set of \emph{terminals} $T \subseteq V(G)$.
We are asked to find an edge set $S$ (called a \emph{Steiner tree}) minimizing $\sum_{e \in S}\omega(e)$ such that every two vertices $u,v \in T$ are connected in the graph $(V,S)$. 
The problem is one of the most important NP-complete problems as it elegantly models network design problems. Several textbooks are entirely devoted to Steiner trees~\cite{Du:2008:STP:1628718, promel2012steiner}.
\paragraph{Parameterization by Number of Terminals}
A very popular research direction that aims to understand the computational complexity of {\pST} is to consider its parameterization by the number of terminals of the instance. Dreyfus and Wagner~\cite{DBLP:journals/networks/DreyfusW71} and independently by Levin~\cite{levin} initiated this line of research and showed that the problem can be solved in $3^{|T|}\poly(n)$ time.\footnote{In this paper we use the $\ohs(\cdot)$ which omits factors polynomial in the input size.}
Thus, in the language of \emph{parameterized complexity}, Dreyfus and Wagner show the problem is Fixed Parameter Tractable when parameterized by $|T|$. 
Fuchs et al.~\cite{DBLP:journals/mst/FuchsKMRRW07} improved this result to $\ohs(c^{|T|})$ for any $c>2$. In the case of small weights, Bj\"orklund et al.~\cite{Bjorklund:2007:FMM:1250790.1250801} provide a faster $\ohs(2^{|T|}\maxweight)$ time algorithm.
All aforementioned algorithms require almost as much working memory as time. However, the setting in which one is given only working memory that is polynomial in the input size has also been well-studied~\cite{DBLP:journals/algorithmica/FominGKLS13, DBLP:conf/icalp/FominKLPS15,DBLP:conf/stoc/LokshtanovN10, DBLP:journals/algorithmica/Nederlof13}. The currently fastest polynomial-space algorithms run in $\ohs(2^{|T|}W)$ time~\cite{DBLP:conf/stoc/LokshtanovN10} and $\ohs(7.97^{|T|})$ time \cite{DBLP:conf/icalp/FominKLPS15}.
\paragraph{Planar Steiner Tree}
Another very popular direction is to study {\pST} restricted to \emph{planar graphs}. The study of approximation schemes for {\pPST} (and many variations and generalizations of it) has been a well-established subject for a long time~\cite{DBLP:conf/soda/BorradaileKK07,DBLP:conf/wads/BorradaileKM07,Bateni:2011:ASS:2027216.2027219,DBLP:conf/soda/BateniCEHKM11}. More recently, our understanding of the exact exponential-time complexity of {\pPST} has also progressed significantly. Some positive results study the decision variant of the unweighted case of {\pPST}, and its parameterization by $|S|$ (the size of the required Steiner tree). Pilipczuk et al.~\cite{DBLP:conf/focs/PilipczukPSL14} show that one can preprocess the input instance in polynomial time to remove all but $O(|S|^{142})$ edges. Pilipczuk et al.~\cite{DBLP:conf/stacs/PilipczukPSL13} (and later, Fomin et al.~\cite{DBLP:conf/focs/FominLMPPS16}) show the problem can be solved in $\ohs(2^{\sqrt{|S| \log^2 |S|}})$ time. The square-root in the exponent is typical for exact algorithms for problems on planar graphs (intuitively, due to the planar separator theorem), and is often called the `square-root phenomenon'. However, such a running time is not always guaranteed. Very recently, it was shown that when parameterized by the number of terminals $|T|$, planarity probably gives little advantage over the algorithm of Dreyfus and Wagner in the following strong sense: if {\pPST} can be solved in $\ohs(2^{o(|T|)})$ time, then the Exponential Time Hypothesis fails~\cite{MarxPP17}.
\paragraph{Planar Steiner Tree with Terminals on Few Faces}
A broadly studied variant of {\pPST} is obtained by making assumptions on the locations of the terminals. Such natural assumptions are also studied extensively in e.g.\ the classic flow paper by Ford and Fulkerson~\cite{Ford-Fulkerson}. Of particular interest is the case when all terminals lie on $k$ given faces of the plane-embedded input graph $G$.
This parameter has a long history in the study of cuts and (multicommodity) flows (e.g.~\cite{DBLP:journals/siamcomp/MatsumotoNS85,DBLP:journals/algorithmica/ChenW03,DBLP:journals/corr/KrauthgamerR17,KrauthgamerLR19}) and shortest paths (e.g.~\cite{Frederickson91,Chen00}).
Krauthgamer et al.~\cite{KrauthgamerLR19} (in this SODA) dubbed it the \emph{terminal face cover number} $\gamma(G)$. The case $\gamma(G) = 1$ is known as an Okamura-Seymour graph~\cite{OkamuraS1981}.
For {\pPST}, the parameterization by $k = \gamma(G)$ generalizes the parameterization by $|T|$, as we can always ensure that $k \leq |T|$. Hence, this parameterization generalizes both previous research directions.

%Note that the parameterization by $k$ generalizes the parameterization by $|T|$, as we can always ensure that $k \leq |T|$, and thus this parameterization generalizes both previous research directions. \textcolor{red}{In addition to its usefulness for {\pPST}, this parameter has also been important in the context of cuts, flows and multicommodity flows\cite{DBLP:journals/siamcomp/MatsumotoNS85,DBLP:journals/algorithmica/ChenW03,DBLP:journals/corr/KrauthgamerR17,KrauthgamerLR19} and shortest paths \cite{Frederickson91,Chen00}.}

%For the {\pPST} problem, a
An important result by Erickson et al.~\cite{DBLP:journals/mor/EricksonMV87} shows that the problem can be solved in $n^{O(k)}$ time. 
%We remark that the special case of their algorithm for 
Their algorithm for $k=1$ arises in both the aforementioned approximation algorithms (i.e.\ in spanner constructions~\cite{DBLP:conf/soda/BorradaileKK07}) and fixed-parameter algorithms (i.e.\ in  preprocessing algorithms~\cite{DBLP:conf/focs/PilipczukPSL14}). Hence, the algorithm plays a central role in the literature on {\pPST}.

%An important result by Erickson et al.~\cite{DBLP:journals/mor/EricksonMV87} in this direction shows that if all terminals lie on $k$ given faces of the plane-embedded input graph, then the problem can be solved in time $n^{O(k)}$. Note that the parameterization by $k$ generalizes the parameterization by $|T|$, as we can always ensure that $k \leq |T|$, and thus this generalizes both previously mentioned research directions. Adding to the importance of the Erickson et al.\ algorithm is that the special case of their algorithm for $k=1$ arises in both the aforementioned approximation algorithms (i.e.\ in spanner constructions~\cite{DBLP:conf/soda/BorradaileKK07}), and in fixed-parameter algorithms (i.e.\ in  preprocessing algorithms~\cite{DBLP:conf/focs/PilipczukPSL14}).

The quest to improve, refine, and generalize the result by Erickson et al.~\cite{DBLP:journals/mor/EricksonMV87} received significant attention. Bern~\cite{thesisbern,doi:10.1002/net.3230200110} improved the constant in the exponent of the running time of~\cite{DBLP:journals/mor/EricksonMV87} to $2$, and gave a better running time if many terminals do not share any face with other terminals. Bern and Bienstock~\cite{DBLP:journals/anor/BernB91} studied a generalization in which the terminals can be removed by removing $k$ consecutive outerplanar layers. Provan~\cite{DBLP:journals/siamcomp/Provan88,DBLP:journals/networks/Provan88} studied generalizations in which covering by faces is replaced with covering by `path-convex regions', motivated by some problems in geometry. For an excellent survey of previous work on {\pPST} with terminals on a few faces we refer to~\cite[Chapter 5]{annalsofDMsteinertree}, or to~\cite{Cheng2000}.

Despite these previous studies going back over 30 years, all previously known algorithms use $n^{\Omega(k)}$ time, and the algorithms matching this time bound use $n^{\Omega(k)}$ space.
From a lower bound perspective, the result by Marx et al.~\cite{MarxPP17} implies that no $\ohs(2^{o(k)})$-time algorithm exists assuming the Exponential Time Hypothesis (as we can always ensure that $k \leq |T|$). However, this still leaves a large gap between the lower and the upper bound. This leads to the natural question what the true computational complexity is of {\pPST} with terminals on $k$ faces.

\subsection{Our Results}
In this work we almost settle the exact complexity of {\pPST} parameterized by the number of faces needed to cover the terminals, modulo the Exponential Time Hypothesis. First, we show that the algorithm of Erickson et al.~\cite{DBLP:journals/mor/EricksonMV87} can be significantly improved:

\begin{theorem}\label{thm:upperbound}
	Given a plane $n$-vertex graph $G$ with terminals $T$, edge weights $\omega:E(G) \rightarrow \{0,\ldots,W\}$ and a set $\cK \subseteq 2^{E(G)}$ of $k$ faces of $G$ such that each vertex from $T$ is on a face in $\cK$, a minimum weight Steiner tree can be found using $2^{O(k)}n^{O(\sqrt{k})}$ time, and polynomial space.
\end{theorem}

Observe that our algorithm uses only polynomial space, in contrast to all previous algorithms with a running time of the type $n^{O(k)}$.

We remark that we may assume that the planar embedding and faces are not a priori given, as already observed in previous work. Explicitly motivated by our setting, Bienstock and Monma~\cite{DBLP:journals/siamcomp/BienstockM88} showed that, given only the graph, one can find $k$ faces covering all terminals in some embedding in $2^{O(k)}\poly(n)$ time, if possible. Hence, we can simply run their algorithm on the input graph before applying Theorem~\ref{thm:upperbound} without affecting the bound on the running time.

We also remark that Marx et al.~\cite{MarxPP17} recently gave an algorithm for {\pPST} with running time $n^{O(\sqrt{|T|})}W$. Note that $k \leq |T|$ and $|T|$ can be arbitrary large when $k=1$, but nevertheless our result is incomparable to theirs because of the $2^{O(k)}$ factor in our running time.

\medskip
We complement our algorithm with a conditional lower bound that almost (that is, modulo the $2^{O(k)}$ factor) matches the running time of our algorithm:

\begin{theorem}\label{thm:lowerbound}
	There is no $f(k)n^{o(\sqrt{k})}$ algorithm for {\pPST} for any computable function $f$, unless the Exponential Time Hypothesis fails.
\end{theorem}
In terms of parameterized complexity, this theorem implies that {\pPST} is $W[1]$-hard when parameterized by the number of terminal faces.

\subsection{Our Techniques}\label{sec:techniques}
We describe our techniques along with intuition and relationship to previous works.

\paragraph{Our Algorithm}
Before we sketch our algorithm, we sketch the previous work we build upon.
A simple observation behind the known exact algorithms for {\pST} (all the way back to~\cite{DBLP:journals/networks/DreyfusW71,levin}) is that any edge $e$ of the solution $S$ splits $S$ into two subtrees $S_1$ and $S_2$; if we know $e$ and which terminal is connected in which subtree, we can simply recursively solve the associated subproblems (or look up their solutions in a Dynamic Programming table). The number of candidates for $e$ and the split of the terminal set is $|E|\cdot 2^{|T|}$, which is (roughly) the running time of~\cite{DBLP:journals/networks/DreyfusW71,levin} and their refinements.

The algorithm by~\cite{DBLP:journals/mor/EricksonMV87} builds upon this scheme and additionally uses the following observation.
Suppose $T' = \{t_1,\ldots,t_p\} \subseteq T$ are terminals that lie on a single face numbered in cyclic order along the face, and $i<j<k<l$ such that $t_i$ and $t_k$ are connected in $S_1$ and $t_j$ and $t_l$ are connected in $S_2$, then $S_1$ and $S_2$ necessarily intersect by planarity and thus we can remove an edge and obtain another solution $S'$ with $\omega(S') \leq \omega(S)$.
Hence, we can restrict our attention to subproblems in which terminal sets form an interval on each face. The number of candidates for $e$ and such terminal sets is $|E|\cdot n^{2k}$, which is (roughly) the running time of~\cite{DBLP:journals/mor/EricksonMV87} and their refinements.

Our approach fits into the general scheme of guessing how separators based on a solution map into an input (see also, for example,~\cite{MarxPP17, DBLP:conf/esa/MarxP15}).
For the aimed running time $2^{O(k)}n^{O(\sqrt{k})}$, we cannot afford the above decomposition as $S_1$ and $S_2$ may interact in $k$ faces from $\cK$. Instead, we use a larger separator on $S$ to decompose $S$ in two forests $S_1$ and $S_2$ such that only a few faces from $\cK$ intersect both $S_1$ and $S_2$. To this end, our \emph{crucial idea} is to consider a separator in the graph $H= S\cup \cK^\flat$, where $\cK^\flat$ denotes the flattening $\cup_{K \in \cK}K$ of $\cK$. We show that $H$ has a (balanced) separator $X$ of size $O(\sqrt{k})$, and if we consider the split of $S$ into $S_1$ and $S_2$ that $X$ induces on $S$, we see that any face in $\cK$ not intersecting $X$ is either entirely connected in $S_1$ or in $S_2$.
Algorithmically, this observation allows us to guess the set $X$, and a partition of the faces from $\cK$ to be covered in both subproblems which we solve recursively.
Faces from $\cK$ intersecting with $X$ can still be connected both via $S_1$ and $S_2$, so their terminals set still needs to be distributed, but by the observation of~\cite{DBLP:journals/mor/EricksonMV87} we can restrict attention to splits induced by intervals.

\paragraph{Our Lower Bound}
Our lower bound builds on ideas of the recent $2^{\Omega(|T|)}\poly(n)$ lower bound by Marx et al.~\cite{MarxPP17}, but instead of reducing from {\pSAT}, we reduce from the \textsc{Grid Tiling} problem. An instance of {\pGT} consists of two integers $n$ and $k$, and $k^2$ sets or \emph{cells} $M_{a,b} \subseteq [n] \times [n]$ for $a,b \in [k]$, and we are tasked to decide whether there exist integers $x_a \in [n]$ and $y_a \in [n]$ for $a \in [k]$ such that $(x_a,y_b) \in M_{a,b}$ for all $a,b \in [k]$. The standard way to do a reduction from {\pGT} in geometric problems (see e.g.~\cite{DBLP:conf/icalp/Marx12}) is to have a gadget for each cell that is capable of representing the choice $(x_a,y_b)$ in that cell, and designing some communication gadget, which when applied to horizontally (or vertically) neighboring cells, ensures that the first (resp. second) index of the choices in these cells are equal. One of the main challenges is to design these communication gadgets, capable of transmitting $\log n$ bits of information between any pair of neighboring cells in the grid of $k\times k$ cells.

Our main innovation in the lower bound is the design of a novel communication gadget, the so-called \emph{flower gadget}, that is capable of communicating multiple bits, but uses only a single terminal face. Essentially, we need a gadget with $2n$ portal vertices with the property that in any optimal solution, the Steiner tree will have exactly two components induced by the gadget, rooted at portal vertex $i$ and $n+i$ respectively for some $i\in \{1,\dots,n\}$. This already prescribes a rotational symmetry to the gadget, but it is challenging to find a gadget that can support this for several reasons. In particular, a pair of `canonical' trees within the gadget must contain an equal number of terminals, and together they must contain all terminals of the terminal face. The easiest way to ensure that the root of one tree uniquely determines the root of the other is to ensure that the root of the tree uniquely determines an interval of terminals from the face that need to be contained in the tree. In practice, this is enforced by making sure that the root of the tree has degree two, and a canonical tree has a (subdivided) binary tree structure. Our gadget is essentially a rolled up Euclidean grid with a special weighting scheme. An important feature of the weighting is that it can be de-weighted: by replacing every edge by a path whose length is the weight of that edge, we obtain an unweighted graph of polynomial size (similarly as in~\cite{MarxPP17}).

\paragraph{Organization}
In Section~\ref{sec:prel} we define some necessary notation and folklore results we will use. 
Section~\ref{sec:upperbound} is devoted to the proof of Theorem~\ref{thm:upperbound}, while Section~\ref{sec:lower} is devoted to the proof of Theorem~\ref{thm:lowerbound}.
%In Section~\ref{sec:steinerbase} we describe a technical ingredient omitted in Section~\ref{sec:prel}.
In Section~\ref{sec:conc} we briefly summarize our paper and point out opportunities for further research.

\section{Preliminaries}\label{sec:prel}

For a set $X$, we let $\Pi(X)$ denote the set of all partitions of $X$. If $\pi \in \Pi(X)$, we write $\pi$ is a partition \emph{of} $X$.
We call a partition $\pi$ \emph{finer} than partition $\pi'$, and denote this relation by $\pi \preceq \pi'$, if for every $u,v \in X$, $u$ and $v$ are in the same block in $\pi$ implies that $u$ and $v$ are in the same block in $\pi'$.
In this case we also say $\pi'$ \emph{coarsens} $\pi$.
Given two partitions $\pi,\pi'$ we use the notation $\pi \sqcup \pi'$ for the join in the partition lattice, that is, the finest partition that coarsens both $\pi$ and $\pi'$.
If $\pi \in \Pi(X)$ and $\pi \in \Pi(Y)$, we also define the join $\pi \sqcup \pi' \in \Pi(X \cup Y)$ obtained by adding singletons to $\pi$ and $\pi'$ to make them elements of $\Pi(X \cup Y)$.
If $u,v \in X$, we write $\{\{u,v\}\}$ for the partition in which all elements except $u$ and $v$ are in singleton blocks.
If $W \subset X$, and $\pi \in \Pi(X)$ we let $\pi_{|W}$ be the projection of $\pi$ on $W$, that is two elements are in the same block of $\pi_{|W}$ if and only if they are in the same block in $\pi$.
If $\cF \subseteq 2^U$, we use the notation $\cF^\flat:=\bigcup_{F \in \cF}F$ for its flattening.

If $G=(V,E)$ is a graph, we write $V(G):=V$ and $E(G):=E$. If $S \subseteq E(G)$, we denote $V(S):=\bigcup_{\{u,v\} \in S}\{u,v\}$ for the set of vertices incident to edges of $S$. For a vertex subset $S\subset V$, let $N[S]$ be the closed neighborhood of $S$, that is, the set $S$ together with all vertices adjacent to a vertex in $S$.
If $X \subseteq E$, we denote $N_X[v]$ to be all vertices sharing an edge in $X$ with $v$.
A vertex subset $X \subseteq V(G)$ is called a \emph{dominating set} if $N[X]=V(G)$.
If $G$ is connected, a vertex $v \in V(G)$ is called an \emph{articulation vertex} if $G[V \setminus \{v\}]$ is not connected.
A graph is \emph{$2$-connected} if it does not contain any articulation vertex.
We say that a path $P$ in $G$ is a \emph{maximal $2$-path} if all internal vertices of $P$ have degree~$2$ and its ends have degree not equal to~$2$. Note that all maximal $2$-paths in $G$ are edge disjoint.

\paragraph{Treewidth and Balanced Separators}
We will use the fact that the treewidth of a plane graph and its dual graph are closely related, which follows from the planar grid minor theorem. In particular, we use the following sharp result:
\begin{theorem}[\cite{BOUCHITTE200134}]\label{thm:dualtw}
	For every plane graph $G$ with dual graph $G^*$, $|\tw(G) - \tw(G^*)| \leq 1$.
\end{theorem}

The following well-known theorem follows in a standard fashion from the grid minor theorem (our particular statement follows from combining Theorem 7.23 from~\cite{DBLP:books/sp/CyganFKLMPPS15} with Lemma 3.1 in~\cite{Demaine:2005:FAK:1077464.1077468}).
\begin{theorem}[\cite{DBLP:journals/jacm/DemaineFHT05}]\label{thm:twvsdomset}
	If a graph $G$ is planar and has a dominating set of size $k$, then $\tw(G)\leq 15\sqrt{k}$.
\end{theorem}

\begin{definition}[Balanced Separation]
	A pair of vertex subsets $(Y,Z)$ is a \emph{separator} in graph $G$ if $Y \cup Z=V(G)$ and there are no edges in $G$ between $Y \setminus Z$ and $Z \setminus Y$. For a fixed weight function $w:V(G) \rightarrow \mathbb{R}$, we say that a separation $(Y,Z)$ is a \emph{$w$-weighted $\alpha$-balanced separation in $G$} if $w(Y\setminus Z)\leq \alpha\cdot w(V(G))$ and $w(Z\setminus Y)\leq \alpha\cdot w(V(G))$.
\end{definition}

\begin{lemma}[Lemma 7.20 from~\cite{DBLP:books/sp/CyganFKLMPPS15}]\label{lem:balcyc}
	Suppose $G$ has treewidth $\tw$, and consider a nonnegative function $w:V(G) \rightarrow \mathbb{R}_{\geq 0}$. Then $G$ has a $\tfrac{2}{3}$-balanced separation $(A,B)$ of order at most $\tw+1$. 
\end{lemma}

\section{Algorithm for {\pPST} with terminals on few faces}
\label{sec:upperbound}

This section is devoted to the proof of Theorem~\ref{thm:upperbound}. Refer to Section~\ref{sec:techniques} for a high level description and intuition. 
To simplify the analysis of our algorithm, we show that w.l.o.g.\ the degree of each vertex is at most~$3$ and the graph is $2$-connected.

\begin{lemma} \label{lem:subcubic}
Let $G$ be a plane graph with terminals~$T$, edge weights $\omega:E(G)\rightarrow \{0,\ldots,W\}$, and a set $\cK \subseteq 2^{E(G)}$ of $k$ faces of $G$ such that each vertex from $T$ is on a face in $\cK$. Then one can compute in polynomial time a $2$-connected subcubic planar graph $G'$ with terminals~$T'$, edge weights $\omega':E(G')\rightarrow \{0,\ldots W\}$, and a set $\cK' \subseteq 2^{E(G')}$ of $k$ faces of $G'$ such that each vertex in $T'$ is on a face in $\cK'$ and moreover, any Steiner tree in $G$ corresponds to a Steiner tree in $G'$ of the same weight and any Steiner tree $S'$ in $G'$ corresponds to a Steiner tree in $G$ of weight at most~$\omega'(S')$.
\end{lemma}
\knip{\begin{proof}
We may assume that $G$ is connected. Otherwise, we can either restrict $G$ to a single component or output a trivial no-instance. We also assume that each terminal $u$ has a unique face $K(u) \in \cK$ such that $u \in V(K)$. We now construct a graph $G''$ to help find $G'$.

Let $G''$ be obtained from $G$ by performing the following procedure simultaneously on all vertices of $G$ of degree at least~$3$. Let $u \in V(G)$ be such a vertex with neighbors $v_1,\ldots,v_\ell$. Replace $u$ by a cycle $C_u$ of length $\ell$ and redirect, for each $i \in [\ell]$, the edge between $u$ and $v_i$ to the $i$-th vertex of $C_u$. This procedure makes $G''$ subcubic and implies a natural embedding of $G''$. For each $u \in T$, observe that $K(u)$ essentially survives in this embedding. We then let $t_u$ be one of the two vertices of $C_u$ on $K(u)$, and let $T'' = \{t_u \mid u \in T\}$ and $\cK'' = \cK$. Let $\omega''(e) = \omega(e)$ for each $e \in E(G') \cap E(G)$ and $\omega''(e) = 0$ for all edges of the cycles $C_u$. Note that any Steiner tree $S$ in $G$ corresponds to a Steiner tree $S''$ in $G''$ such that $\omega''(S'') = \omega(S)$ by adding to $S$, for each $u \in S$, all but one of the edges of $C_u$. Conversely, any Steiner tree $S''$ in $G''$ corresponds to a Steiner tree $S$ in $G$ by removing from $S''$ all edges in $C_u$ and taking a spanning tree; then $\omega(S) \leq \omega''(S'')$.

We observe that the removal of any articulation vertex $v$ of $G''$ of degree~$3$ leads to exactly two connected components, meaning that $v$ is incident on a bridge (and a maximal $2$-path).

Let $G'$ be obtained from $G''$ by performing the following procedure simultaneously on all maximal $2$-paths of $G''$. Let $P$ be such a path with ends $u,v$. Note that all edges of $P$ border two faces $K, K'$ (possibly $K = K'$). If $V(P) \cap T'' \not= \emptyset$, then without loss of generality $K = K(t)$ for all $t \in V(P) \cap T''$. Create copies $u'$ and $v'$ of $u$ and $v$ respectively, make $u$ adjacent to $u'$ and $v$ to $v'$, and set the weight of the new edges to $0$. If $u$ has degree~$3$, then redirect one of the incident edges of $u$ not on $P$ to $u'$. Do the same for $v$ and $v'$. Now add a path $P'$ from $u'$ to $v'$ of the same length as $P$ such that all new edges receive weight $W$. Let $G'$ be the resulting graph. Observe that $G'$ may be embedded in the plane such that $P'$ borders $K'$, $P$ still borders $K$, and $P$, $P'$, $\{u,u'\}$, and $\{v,v'\}$ jointly form a new face; hence, $G'$ is planar. Moreover, every terminal is still on a face in $\cK''$. Let $\omega'(e) = \omega''(e)$ for all $e \in E(G') \cap E(G'')$ and let $\omega'(e)$ be as described otherwise. Let $T' = T''$ and $\cK' = \cK''$. Note that any Steiner tree $S''$ in $G''$ corresponds to a Steiner tree $S'$ in $G'$ such that $\omega'(S') = \omega''(S'')$ by adding to $S'$ all edges $\{u,u'\}, \{v,v'\}$. Conversely, any Steiner tree $S'$ in $G'$ corresponds to a Steiner tree $S''$ in $G''$ by replacing in $S''$ all edges of the paths $P'$ with corresponding edges in $P$ and taking a spanning tree; then $\omega''(S'') \leq \omega'(S')$.
\end{proof}}

We are now ready to describe the algorithm. Since we present a recursive algorithm, it is more convenient to work with a slightly more general problem that is solved in recursive steps. We first define this more general problem.
 
\begin{definition}[Block Steiner Forest]\label{def:bsf}
Given nonempty subsets $B,T \subseteq V(G)$, and a partition $\pi$ of $B$, we say $S \subseteq E(G)$ is a \emph{$(G,B,\pi,T)$-Block Steiner Forest} if in $(V(G),S)$
	\begin{enumerate}[(a)]
		\item every vertex in $T$ is connected to at least one vertex in $B$, and
		\item two vertices in $B$ are connected if and only if they are in the same block of $\pi$.
	\end{enumerate}
\end{definition}
% \skb{The formatting here is a bit weird, because the indentation for the first line of a paragraph is as long the offset for the problem definition}
\optimization{\pPBSF{} ({\sc PBSF})}{Plane $G$, weights $\omega:E(G)\rightarrow \{0,\ldots,W\}$, subsets $B,T \subseteq V(G)$, partition $\pi$ of $B$.}{The minimum $\omega(S)$ where $S$ is a $(G,B,\pi,T)$-Block Steiner Forest $S$.}

\noindent In the above problem we can think of the vertices in $B$ as \emph{boundary vertices}. Typically, in the recursion other parts of the Steiner tree will intersect only in $B$ and already establish some connectivity, which allows us to only connect vertices in $B$ according to $\pi$.
We need to establish the following for the base case of the recursive algorithm.

\begin{lemma}\label{lem:basecase}
	For any constant $c_0$ there is an algorithm $\mathtt{steinerBase}(G,\omega,B,\pi,T,\cK)$ that solves a Block Steiner Forest instance $(G,B,\pi,T)$ in polynomial time, provided that $|B|+|\cK|\leq c_0$.
\end{lemma}

% \section{Implementation of Algorithm $\mathtt{steinerBase}$}
% \label{sec:steinerbase}
% In this section we give the postponed proof of Lemma~\ref{lem:basecase}.

Before we prove Lemma~\ref{lem:basecase}, we need to introduce the following tools.

\begin{theorem}[\cite{DBLP:journals/mor/EricksonMV87}]\label{thm:erickson}
	Let $(G,\omega,T)$ be a given instance of {\pPST}, and let $\cK$ be a given set of $O(1)$ faces such that $T \subseteq V(\cK)$. Then the instance can be solved in polynomial time.
\end{theorem}

\begin{definition}[Non-crossing sequence]
	A sequence $x \in [\ell]^n$ is \emph{non-crossing} if for all $1 \leq i<j<k<l \leq n$, we have that $x_i\neq x_k$, $x_j \neq x_l$, or $x_i,x_j,x_k,x_l$ are all equal. It is \emph{minimal} if there is no $i$ such that $x_i=x_{i+1}=x_{i+2}$.
\end{definition}

\begin{lemma} \label{lem:seqlength}
	The length of a minimal non-crossing sequence with $\ell$ different values is at most $4\ell$.
\end{lemma}
\begin{proof}
	Use induction on $\ell$. For $\ell=1$ the statement trivially holds, so assume $\ell > 1$.
	For a value $v \in [\ell]$, denote $l(v)$ (and $r(v)$) the smallest $i$ (and largest) $i$ such that $x_i=v$. As any two intervals $[l(v),r(v)]$ and $[l(v'),r(v')]$ are either disjoint or contained in each other, we can always find an interval that does not contain any other interval. Then necessarily $r(v)=l(v)+1$ or $r(v)=l(v)$. Consider the sequence obtained by removing indices $l(v),r(v)$. If this leads to a quadruple or triple of consecutive equal values, then remove at most two of them so that exactly two remain. We obtain a minimal non-crossing sequence on $\ell-1$ values which has length at most $4(\ell-1)$ by induction. Since we have removed at most four indices, the lemma follows.
\end{proof}

\begin{proof}[Proof of Lemma~\ref{lem:basecase}]
%	As before, we assume $G$ is $2$-connected.
%	Suppose $(G,\omega,B,T,\pi)$ is an instance of Planar Block Steiner Forest, and $\cK$ is a set of faces of $G$ such that $T \subseteq V(\cK)$ with $|\cK|$ being a constant.
Let $S \subseteq E(G)$ be an optimal solution. Then there exists a partition $\pi_S$ of $B$ such that $\pi \preceq \pi_S$ and two vertices $u,v$ of $B$ are in the same connected component of $G[S]$ if and only if $\{\{u,v\}\} \preceq \pi_S$. By enumerating all possibilities, we may by abuse of notation assume that $\pi = \pi_S$. This adds a constant factor to the running time, as $|B| \leq c_0$.
%	We know an optimal solution $S \subseteq E(G)$ consists of connected components each connecting a unique block of $\pi$, as each connected component is required to include some vertices of $B$ and if a connected component contains two blocks of $\pi$ we could remove an edge to get a solution that is not more expensive.
	
	Let $\{B_1,\ldots,B_\ell\}$ be the partition $\pi$, and let $S_1,\ldots,S_\ell$ be the corresponding subtrees of $S$. Let $K \in \cK$ have terminals $t_1,\ldots,t_p$, enumerated in order of appearance of a walk on the face (since $G$ is $2$-connected, the face boundary forms a cycle). Suppose that $i<j<k<l$, and terminals $t_i$ and $t_k$ are connected in $S$ and terminals $t_j$ and $t_l$ are connected in $S$. Then all four terminals must be connected to each other in a tree $S_z$ as they all lie on the same face. Thus, if for each terminal in $t_i$ we let $x_i \in [\ell]$ encode the index of the block it is connected to within $S$, then $x$ is a non-crossing sequence.
	
	Such a non-crossing sequence can be encoded by its minimal non-crossing sequence (obtained by removing all but two elements from each subsequence of the same element) and a mapping from the indices of the non-crossing sequence to $V(K) \cap T$. As the length of the minimal non-crossing sequence is at most $4\ell$ by Lemma~\ref{lem:seqlength}, there are at most $\ell^{4\ell} n^{4\ell}$ different sequences~$x$.
	
	The algorithm now is as follows: enumerate all possible combinations of sequences $x_K\in [\ell]^{|V(K) \cap T|}$ for each face $K \in \cK$.
	Then for each $i \in [\ell]$ solve the instance $(G,\omega,T_i)$ using the algorithm of Theorem~\ref{thm:erickson}, where $T_i=\bigcup_{K \in \cK}\, x^{-1}_K(i)$.
	The running time is polynomial since $\ell \leq |B|$ and $|\cK|$ are constants.
\end{proof}

Our main effort in the remainder of this section will be to prove the following lemma, of which Theorem~\ref{thm:upperbound} is an easy consequence (we postpone the proof of Theorem~\ref{thm:upperbound} to the end of this section). By Lemma~\ref{lem:subcubic}, we may assume the input graph is $2$-connected and subcubic. 

\begin{lemma}\label{lem:steineralgo}
	Suppose $(G,\omega,B,T,\pi)$ is an instance of Planar Block Steiner Forest, and $\cK$ is a set of faces of $G$ such that $T \subseteq V(\cK)$.
	Then Algorithm $\mathtt{steiner}(G,\omega,B,\pi,T,\cK)$ as listed in Algorithm~\ref{alg:steiner} correctly solves the {\sc PBSF} instance $(G,B,\pi,T)$.
\end{lemma}

We continue with the description of the procedure $\mathtt{steiner}$ as listed in Algorithm~\ref{alg:steiner}.
%It assumes $G$ is $2$-connected. Note that if $G$ is not $2$-connected, $v$ is an articulation vertex and $G[V \setminus \{v\}]$ consists of connected components $C_1,\ldots,C_l$, then we can recurse on the subinstances $G[C_i \cup \{v\}]$ in a straightforward fashion.
For a subset $X \subseteq V(G)$, we let $\cK(X) \subseteq \cK$ be the set of faces from $\cK$ whose edges intersect with $X$.
For a face $K \in \cK$ and a vertex set $X \subseteq V(G)$, let $\cc(K,X) \in \Pi(V(K)\setminus X)$ be the partition of the face vertices induced by removing $X$, that is, $\cc(K,X)$ is the collection of vertex sets of the connected components of the subgraph of $(V(K),E(K))$ induced by $V(K) \setminus X$.

In Line~\ref{lin:base} we use Lemma~\ref{lem:basecase} as the base case.
At Line~\ref{line:guesssep} we guess what the separator $X$ is (as already described in Section~\ref{sec:techniques}), and at Line~\ref{line:guesssplit} we guess how the boundary vertices and faces not intersecting $X$ are distributed among the subproblems.
%Using the loop starting at Line~\ref{lin:makefacesdisjoint} we ensure $V(K)$ and $V(K')$ are disjoint for distinct $K,K' \in \cK$ for technical convenience. It is easy to see that $\cK$ remains a set of faces of the new graph.
At Line~\ref{lin:splitterminalfaces} we guess for each segment of faces from $\cK$ obtained after removing $X$ whether they are connected in the first or second subproblem.
Note that these segments do not include $X$ itself.
Based on all these guesses we compute the set of terminals $T_1$ and $T_2$ to be connected in both subproblems on Line~\ref{lin:T1} and Line~\ref{lin:T2}. At Line~\ref{lin:looppartitions} we guess what connectivity is established in both problems (encoded as partitions), and on Line~\ref{lin:respart} we check whether the two partitions jointly encode all required connectivity.

\begin{algorithm}[t]
	\caption{Algorithm implementing Theorem~\ref{thm:upperbound}.}
	\label{alg:steiner}
	\begin{algorithmic}[1]
		\REQUIRE $\mathtt{steiner}(G,\omega,B,\pi,T,\cK)$ \hfill\algcomment{Assumes $T \subseteq V(\cK)$, $G$ is $2$-connected}
		\ENSURE Minimum $\omega(S)$ over all $(G,B,\pi,T)$-Block Steiner Forests $S\subseteq E(G)$.
		\STATE \textbf{If} $|B|+|\cK(T)| \leq c_0$ \textbf{then} \textbf{return} $\mathtt{steinerBase}(G,\omega,B,\pi,T,\cK)$\label{lin:base}
		\STATE $\best \gets \infty$
		\FOR{every $\displaystyle X \in \binom{V(G)}{\leq 15\sqrt{|\cK(T)|+|B|}+2}$}\label{line:guesssep}
			\FOR{$B_1 \subseteq B \setminus X$, and $\cK_1 \subseteq \cK(T) \setminus \cK(X)$ such that $\displaystyle\frac{|B_1|+|\cK_1|}{|B|+|\cK(T)|} \in [\tfrac{1}{3},\tfrac{2}{3}]$}\label{line:guesssplit}
				\item[] \algcomment{Based on the above guessed split of boundary vertices and terminal faces for}
				\item[] \algcomment{the first subproblem, compute the corresponding sets for the second subproblem}
				\STATE $B_2 \gets (B \setminus X)\setminus B_1$
				\STATE $\cK_2 \gets (\cK(T) \setminus \cK(X))  \setminus \cK_1$
%				\item[] \algcomment{For technical convenience, ensure $V(K)\cap V(K')\neq \emptyset$ implies $K=K'$ by adding edges }
%				\WHILE{exists $v \in V(G)$ and distinct $K,K' \in \cK$ such that $v \in V(K) \cap K'$}\label{lin:makefacesdisjoint}
%					\STATE let $N_{K}(v)=\{u,w\}$, add $\{u,w\}$ to $G$ and $K$; remove $\{u,v\},\{v,w\}$ from $K$
%				\ENDWHILE
				\item[] \algcomment{Try all subsets of segments of terminal faces to assign to the first subproblem}
				\FOR{all $\displaystyle \cA_1 \subseteq \bigcup_{K \in \cK}\cc(K,X)$} \label{lin:splitterminalfaces}
					\item[] \algcomment{Compute the terminal sets for both subproblems based on the above guesses}
					\STATE $\displaystyle T_1 \gets (\cA_1^\flat \cap T) \cup \bigcup_{K \in \cK_1} (V(K) \cap T) $ \label{lin:T1}
					\STATE $\displaystyle T_2 \gets (V(\cK(X)^\flat) \setminus X) \setminus \cA_1^\flat \cup \bigcup_{K \in \cK_2} (V(K) \cap T)$ \label{lin:T2}
					\FOR{all partitions $\pi_1$ on $B_1 \cup X$ and partitions $\pi_2$ on $B_2 \cup X$}\label{lin:looppartitions}
						\STATE $\pi' \gets \pi_1 \sqcup \pi_2$
						\item[] \algcomment{Check if the two partitions implement the required connectivity}
						\IF{$\pi'_{|B} = \pi$ \textbf{and} for all $u \in X$, there exists $v \in B$ with $ \{ \{ u,v \} \} \preceq \pi'$}\label{lin:respart}
							\item[] \algcomment{Solve the subproblems recursively, and update current minimum if needed}
							\STATE $\best_1 \gets \mathtt{steiner}(G,\omega,B_{1} \cup X, \pi_1,T_1 \setminus X,\cK_1 \cup \cK(X))$\label{lin:sub1}
							\STATE $\best_2 \gets \mathtt{steiner}(G,\omega,B_{2} \cup X, \pi_2,T_2 \setminus X,\cK_2 \cup \cK(X))$\label{lin:sub2}
							\STATE $\best \gets \min\{\best,\best_1+\best_2\}$\label{lin:update}
						\ENDIF
					\ENDFOR
				\ENDFOR
			\ENDFOR	
		\ENDFOR
		\STATE \algorithmicreturn\ $\best$
	\end{algorithmic}
\end{algorithm}

\begin{proof}[Proof of Lemma~\ref{lem:steineralgo}]
 We need to establish that the algorithm gives a feasible solution that is optimal, which we do in two steps.

\paragraph{Correctness: Feasibility}
Let $\best:=\mathtt{steiner}(G,\omega,B,\pi,T,\cK)$.
We prove that $\omega(S) \leq \best$ for some $(G,B,\pi,T)$-Block Steiner Forest $S$.
The base case follows from Lemma~\ref{lem:basecase}.
For the recursive case, consider the iteration at Line~\ref{lin:update} where $\best$ is updated for the last time.
By induction there is a $(G,B_1 \cup X,\pi_1,T_1)$-Block Steiner Forest $S_1$ such that $\omega(S_1) \leq \best_1$, and a $(G,B_2 \cup X,\pi_2,T_2)$-Block Steiner Forest $S_2$ such that $\omega(S_2) \leq \best_2$.

We claim that $S:=S_1\cup S_2$ is a $(G,B,\pi,T)$-Block Steiner Forest.
Note that two vertices in $B_1 \cup B_2 \cup X$ are connected in $S$ if and only if they are in the same block of $\pi'=\pi_1 \sqcup \pi_2$.
Therefore, as we require that $(\pi_1 \sqcup \pi_2)_{|B}=\pi$ on Line~\ref{lin:respart}, $S$ satisfies property $(b)$ of Definition~\ref{def:bsf}.
To see that $S$ also satisfies $(a)$ of Definition~\ref{def:bsf}, consider some terminal $t \in T$. We distinguish three cases:

\begin{itemize}
	\item If $t=u \in X$, then $u$ is connected to a vertex $v \in B$ as we require $\{ \{u,v\} \} \preceq \pi_1 \sqcup \pi_2$ for some $v \in B$ on Line~\ref{lin:respart}.
	\item If $t \notin X$, and $\cK({t}) \in \cK(X)$, then $t$ will be either in $T_i$ for $i=1$ or $i=2$, depending on whether the member of $\cc(K,X)$ containing $t$ is in $\cA_1$ or not, and by induction $t$ will be connected to some vertex $u \in B_i \cup X$ in $(V,S_i)$. If $t$ is connected to a vertex in $B_i$, then we are done as $B_i \subseteq B$; if $t$ is connected to a vertex in $X$, then the first case applies.
	\item If $t \notin X$, and $\cK({t}) \in \cK_i$ for $i \in \{1,2\}$ then either at Line~\ref{lin:T1} or Line~\ref{lin:T2} we add $t$ to $T_i$, and by induction $t$ will be connected to some vertex in $B_i \cup X$ in $(V,S_i)$. If $t$ is connected to a vertex in $B_i$, then we are done as $B_i \subseteq B$; if $t$ is connected to a vertex in $X$, then the first case applies.
\end{itemize}

Thus $S$ is a $(G,B,\pi,T)$-Block Steiner Forest. The claim follows as
\[
	\omega(S) \leq \omega(S_1)+\omega(S_2) \leq \best_1+\best_2 \leq \best.
\]
\paragraph{Correctness: Optimality}
Denote $\best:=\mathtt{steiner}(G,\omega,B,\pi,T,\cK)$.
We prove that $\omega(S) \geq \best$ for every $(G,B,\pi,T)$-Block Steiner Forest $S$.
We do this by showing that there exists some partition $S_1$, $S_2$ of $S$ such that in some iteration $\omega(S_1) \geq \best_1$ and $\omega(S_2) \geq \best_2$.
Note that since $\omega(e)\geq 0$ for every $e \in E(G)$, we may assume that $S$ is a forest: if $S$ would have a cycle, then we could remove any edge on the cycle to obtain a new $S$ with less or equal weight.

Consider the subgraph $H=(V(S)\cup V(\cK(T)) \cup B,S \cup \cK(T)^\flat)$ of $G$, with the embedding inherited from the embedding of $G$. Let $H^*$ be the planar dual of $H$; that is, for every face of $H$ we create a vertex $H^*$, and two vertices in $H^*$ are connected with an edge in $H^*$ if and only if the corresponding faces in $H$ share an edge.

\begin{claim}\label{clm:twbound}
	$\tw(H) \leq 15\sqrt{|\cK(T)|}+1$.
\end{claim}
\begin{proof}
If $\cK(T)=\emptyset$, then $H$ is a forest and $\tw(H)=1$.
Otherwise, we claim that $\cK(T) \subseteq V(H^*)$ is a dominating set of $H^*$.
To see this, note that removing the edges of a face in the primal $H$ amounts to contracting the neighborhood of the corresponding vertex in the dual to a single vertex.
Therefore we know that if we contract the sets $N_{H^*}[v]$ for all $v \in \cK(T)$, we are left with a single vertex (being the planar dual of a forest).
This implies that $\cK(T)$ is a dominating set of $H^*$: if there is a vertex in $V(H^*)\setminus N_{H^*}[\cK]$, it would still be a vertex in the graph after contracting and there would be at least two vertices in the planar dual of the forest, which is a contradiction.

Now we use Lemma~\ref{thm:twvsdomset} to obtain $\tw(H^*) = 15\sqrt{|\cK(T)|}$. As $H^*$ is the dual of $H$, by Theorem~\ref{thm:dualtw} we have the bound $\tw(H)\leq 15\sqrt{|\cK(T)|}+1$. 
\end{proof}
%\skb{This is a decent place to remark that this treewidth bound is tight up to constant factors; we can just reference the canonical solution of our own lower bound construction.}\todo{J: IMO this distracts from the proof, and it's quite easy to see as the grid (with all faces in $\cK$) is already an example}
Now we consider the following weight function $w: V(H) \rightarrow \mathbb{N}$: 
\begin{itemize}
	\item For every $v \in V(G)$, set $w(v)=0$
	\item For every face $K$ in $\cK(T)$, arbitrarily pick a vertex $v \in V(K)$ (which could be in $B$) and set $w(v)=1$,
	\item For every $v \in B$, set $w(v)= w(v)+1$.
\end{itemize}

By Claim~\ref{clm:twbound} and Lemma~\ref{lem:balcyc} there is a $w$-weighted $\tfrac{2}{3}$-balanced separation $(Y,Z)$ in $H$ such that $|Y\cap Z| \leq 15\sqrt{k}+2$. 
In some iteration of the loop at Line~\ref{line:guesssep}, the algorithm will set $X=Y \cap Z$.
Since the separation $(Y,Z)$ of $H$ is balanced with respect to $w$, we have that $w(Y),w(Z)\leq \tfrac{2}{3}w(V(G))$. This implies that $\frac{|B \cap (Y \setminus X)|+|\cK(Y) \setminus \cK(X)|}{|\cK(T)|+|B|} \in [\tfrac{1}{3},\tfrac{2}{3}]$. Therefore, the algorithm will set $B_1=(B \cap Y) \setminus X$ and $\cK_1=\cK(Y)\setminus\cK(X)$ in some iteration of the loop at Line~\ref{line:guesssplit}.

Note that in this iteration we set $B_2 = (B \setminus X)\setminus B_1$ which equals $(B \cap Z) \setminus X$, since $B \subseteq Y \cup Z$. 
Similarly, we set $\cK_2=(\cK(T) \setminus \cK(X))\setminus \cK_1$.

Note that $\cK_2=\cK(Z)\setminus \cK(X)$, because if a face has vertices from both $Y$ and $Z$, then it must also have vertices from $X$ as $(X,Y)$ is a separation and every face is a cycle (this follows as $G$ is $2$-connected, see i.e.~\cite[Theorem 4.2.5]{DBLP:books/daglib/0030488}). Moreover, if two vertices $a,b \notin X$ are in the same connected component $C \in \cc(K,X)$, then either both are in $Y$ or both are in $Z$, as $(Y,Z)$ is a separation of a graph containing the edge set $K$. 

Thus, at some iteration of the loop at Line~\ref{lin:splitterminalfaces} the algorithm will set $\cA_1$ such that for every face $K\in\cK(X)$, any connected component $C$ of $\cc(K,X)$ is contained in $Y$ if it is in $\cA_1^\flat \cap V(K)$, and it is contained in $Z$ otherwise. It follows that there is some iteration in which $T_1 \setminus X = (Y \setminus X) \cap T$ and $T_2 \setminus X= (Z \setminus X) \cap T$.

The separation $(Y,Z)$ of $H$ induces a partition of $S$ into two subforests $S_1,S_2$, where $S_1= \{ \{u,v\} \in S: \{u,v\} \in Y \}$ and $S_2:= S \setminus S_1$ (note that we add edges of $S$ contained in $E(X)$ to $S_1$ and not to $S_2$).

Let $\pi_1$ be the partition on $B_1 \cup X$ where $\{u,v\} \in B_1 \cup X$ are in the same block of $\pi_1$ if and only if $u$ and $v$ are connected in the graph $(V,S_1)$.
Similarly, let $\pi_2$ be the partition on $B_2 \cup X$ where $\{u,v\} \in B_2 \cup X$ are in the same block of $\pi_2$ if and only if $u$ and $v$ are connected in the graph $(V,S_2)$.

Since $S=S_1 \cup S_2$ is a $(G,B,\pi,T)$-Block Steiner Forest, we see that $\pi_1 \sqcup \pi_2$ equals $\pi$, and that it satisfies the conditions checked on Line~\ref{lin:respart}.
As the algorithm loops over all options of $\pi_1,\pi_2$ on Line~\ref{lin:looppartitions}, eventually it will try the pair $\pi_1,\pi_2$.

We can conclude that $S_1$ is a $(G,B_1 \cup X,\pi_1,T_1)$-Block Steiner Tree, and that $S_2$ is a $(G,B_2 \cup X,\pi_2,T_2)$-Block Steiner Tree.
As $\cK_1 \subseteq \cK(T)$ and $\cK_2 \subseteq \cK(T)$ we have by induction that $\omega(S_1) \geq \best_1$ and, $\omega(S_2) \geq \best_2$. Now $\omega(S)=\omega(S_1)+\omega(S_2)\geq \best_1 + \best_2 \geq \best$ follows. This concludes the proof.
\end{proof}

\begin{proof}[Proof of Theorem~\ref{thm:upperbound}]
	Arbitrarily pick a terminal $t_0 \in T$. Then $\mathtt{steiner}(G,\omega,\{t_0\},\{\{t_0\}\},T_0\setminus t_0,\cK)$ will be the minimum weight of a tree connecting $T$ by Lemma~\ref{lem:steineralgo}, which is exactly what needs to be computed in the {\pPST} instance.
	Since $G$ is subcubic, every vertex is in at most three faces of $\cK$ and thus, $|\cK(X)|+|X|\leq 4|X|$. 
	Therefore, if $|\cK|+|B|$ is larger than some constant $c_0$, then in a recursive call with parameters $\cK',B'$ we have
	\[
		|\cK'|+|B'|\leq \tfrac{2}{3}(|\cK|+|B|)+4(15\sqrt{|\cK|+|B|}+2)\leq \tfrac{3}{4}(|\cK|+|B|).
	\]
	Thus the recursion depth of $\mathtt{steiner}$ is at most $O(\log |\cK|)$, and $|B| = O(\sqrt{k}\log k)$ for any recursive call.
	
	If we let $T(n,p)$ denote the running time of $\mathtt{steiner}$ when $|B|+|\cK|=p$ we see that
	\[ T(n,p) =
	\begin{cases}
		n^{O(1)} 									& \text{ if $p$ is constant and $\mathtt{steinerBase}$ is called} \\
		n^{O(\sqrt{p})}2^{O(p)} T(n,\tfrac{3}{4}p)  & \text{otherwise.}
	\end{cases}
	\]
	To see this, note the loop at Line~\ref{line:guesssep} has at most $n^{O(|X|)}$ iterations; the loop at Line~\ref{line:guesssplit} has at most $2^{|B|+|\cK|}$ iterations, the loop on Line~\ref{lin:splitterminalfaces} has at most $2^{|X|}$ iterations (as $|\cc(K,X)|\leq |X|$), and the loop on Line~\ref{lin:looppartitions} has at most $(|B|+|X|)^{O(|B|+|X|)}$ iterations (as $|\pi(B)| =|B|^{O(|B|)}$). All other operations (apart from the recursion) require polynomial time.
	This recurrence solves to $2^{O(p)}n^{O(\sqrt{p})}$, thus the theorem follows.	
\end{proof}

\section{Lower Bound}\label{sec:lower}
In this section, we aim to prove Theorem~\ref{thm:lowerbound}. Throughout, for any integer $n$, let $[n] = \{1,\ldots,n\}$ (where $[0] = \emptyset$) and let $[n] \times [n] = \{(1,1),\ldots,(1,n), (2,1),\ldots,(n,n)\}$.

We present a reduction from {\pGT}, which is defined as follows. An instance of {\pGT} consists of two integers $n$ and $k$, and $k^2$ sets $M_{a,b} \subseteq [n] \times [n]$ for $a,b \in [k]$. Let $\mathcal{M} = \{M_{a,b} \mid a,b \in [k]\}$. Since $n$ and $k$ can be derived by inspecting $\mathcal{M}$, we may specify the instance by $\mathcal{M}$ alone. The {\pGT} problem asks to decide whether there exist integers $x_a \in [n]$ and $y_a \in [n]$ for $a \in [k]$ such that $(x_a,y_b) \in M_{a,b}$ for all $a,b \in [k]$. In this case, we call $x_1,\ldots,x_k,y_1,\ldots,y_k$ a solution to the instance. The following statement is known for {\pGT}.

\begin{theorem}[\cite{DBLP:conf/icalp/Marx12,DBLP:books/sp/CyganFKLMPPS15}] \label{thm:g:gt}
There is no $f(k) \cdot n^{o(k)}$-time algorithm for {\pGT} for any computable function $f$, unless the Exponential Time Hypothesis fails.
\end{theorem}
The reduction is implied by the following theorem.

\begin{theorem} \label{thm:g:bound}
Let $\mathcal{M}$ be an instance of {\pGT}, with associated integers $n$ and $k$. Then in time polynomial in $n$ and $k$, one can construct an integer $K_{\mathcal{M}}$ a planar graph $\mathcal{G}_{\mathcal{M}}$ and set $T_{\mathcal{M}}$ of terminals such that
\begin{itemize}
\item $\mathcal{G}_{\mathcal{M}}$ has size $O(k^{16} n^{27})$;
\item $T_{\mathcal{M}}$ can be covered by $k(k-1)+1$ faces of an embedding of $\mathcal{G}_{\mathcal{M}}$;
\item $\mathcal{M}$ admits a solution if and only if $\mathcal{G}_{\mathcal{M}}$ admits a Steiner tree of at most $K_{\mathcal{M}}$ edges.
\end{itemize}
\end{theorem}

The construction draws on ideas from Marx et al.~\cite{MarxPP17}, but differs in crucial points. A sequence of the elaborate verification gadgets by Marx et al.~\cite{MarxPP17} is capable of communicating while relying on only constantly many terminals as long as the Steiner tree connects these gadgets in a chain. In essence, the gadgetry can be used to represent the choices in a single row of the $k\times k$ grid, and ensure the communication. In addition, each copy of the verification gadget is capable of extracting $1$ bit of information vertically. Unfortunately such gadgetry cannot be used for communicating both horizontally and vertically in the $k\times k$ grid, since the chaining property would mean that the Steiner tree would have to induce cycles. To get around this problem, Marx et al. designs a ``connector gadget'' that can transmit one bit of information vertically, without introducing connectivity. The gadget uses four terminals adjacent that are on one face. Since we would need a large number of terminal faces to communicate the required bits vertically, the connector gadget does not yield a parameterized reduction for our parameter $k$. In order to extract multiple bits, we modify the gadget sequence into $n$-tuples of verification gadgets; this only leaves open the issue of communicating $\omega(1)$ bits of information without connectivity, and using only $O(1)$ terminal faces.

Given Theorem~\ref{thm:g:bound}, Theorem~\ref{thm:lowerbound} is quickly proven.

\begin{proof}[Proof of Theorem~\ref{thm:lowerbound}]
Suppose there is an $f(k)n^{o(\sqrt{k})}$-time algorithm $\mathcal{A}$ for {\pPST} for some computable function $f$. We now construct a fast algorithm for {\pGT}. Let $\mathcal{M}$ be an instance of {\pGT}, with associated integers $n$ and $k'$. Apply Theorem~\ref{thm:g:bound} to $\mathcal{M}$, which takes time polynomial in $n$ and $k'$, and let $K_{\mathcal{M}}$ be the resulting integer, $\mathcal{G}_{\mathcal{M}}$ the resulting planar graph, and $T_{\mathcal{M}}$ the corresponding terminal set. Run $\mathcal{A}$ on $\mathcal{G}_{\mathcal{M}}$ and let $S$ denote the resulting Steiner tree. Answer ``yes'' if $S$ has at most $K_{\mathcal{M}}$ edges, and answer ``no'' otherwise. This completes the description of the algorithm.

The correctness of the algorithm is immediate from the third item of Theorem~\ref{thm:g:bound}. Since $T_{\mathcal{M}}$ can be covered by $k'(k'-1)+1$ faces of an embedding of $\mathcal{G}_{\mathcal{M}}$ and $\mathcal{G}_{\mathcal{M}}$ has size $O((k')^{16} n^{27})$, the algorithm runs in $f'(k') n^{o(k')}$ time for some computable function $f'$. According to Theorem~\ref{thm:g:gt}, this implies that the Exponential Time Hypothesis fails.
\end{proof}

We now set out to prove Theorem~\ref{thm:g:bound}. Throughout, let $N, L, t$ be large integers to be chosen later. Let $M$ be a large integer such that $M > 10\cdot NL$ and let $M_i = M^i$. The construction consists of two types of gadgets. The first is the flower gadget (the main novelty of our construction), the second are the verification gadgets. We now present both types in detail and discuss their properties. Then we show how these gadgets can be brought together and prove Theorem~\ref{thm:g:bound}.

%\subsection{Flower Gadget}

\subsection{Flower Gadget}

The gadget is based on a ``rolled'' grid whose edges have a special weighting scheme. It is easier to study the unrolled version of the grid first, which we do in Section~\ref{sec:unrolled}. We establish some properties of the metric space induced by the weighting, and prove the key statement (Lemma~\ref{lem:triangle}) that we need for Steiner trees in this unrolled context. Then, in Section~\ref{sec:rolled}, we present the full rolled gadget and prove the essential properties we need in the final construction.

\subsubsection{The unrolled grid} \label{sec:unrolled}

For $a,b \in \Ints$, the discrete interval with endpoints $a,b$ is the set $\{a,a+1,\dots,b\}$, which we denote by $\lb a,b\rb$. The discrete intervals form a poset with respect to the containment relation: let $\Gamma$ be the (undirected) Hasse diagram of this poset. Equivalently, $\Gamma$ is the subgraph of the square grid restricted to integer points on or above the line $y-x=0$. In this diagram, we can talk about ancestors and descendants; in particular, for any pair of points $p$ and $q$ the lowest common ancestor (the smallest interval that contains both $p$ and $q$) is well-defined.

We define a weight function $w$ on the edges of $\Gamma$ in the following way. For an edge $pq$, let $q=\lb a,b\rb$ be the larger discrete interval, that is, suppose $p\subset q$. Then the weight $w(pq)$ is set to $2^{-\lfloor\log(b-a)\rfloor}$. See Figure~\ref{fig:gamma}.

\begin{figure}
\begin{center}
\includegraphics[scale=0.75]{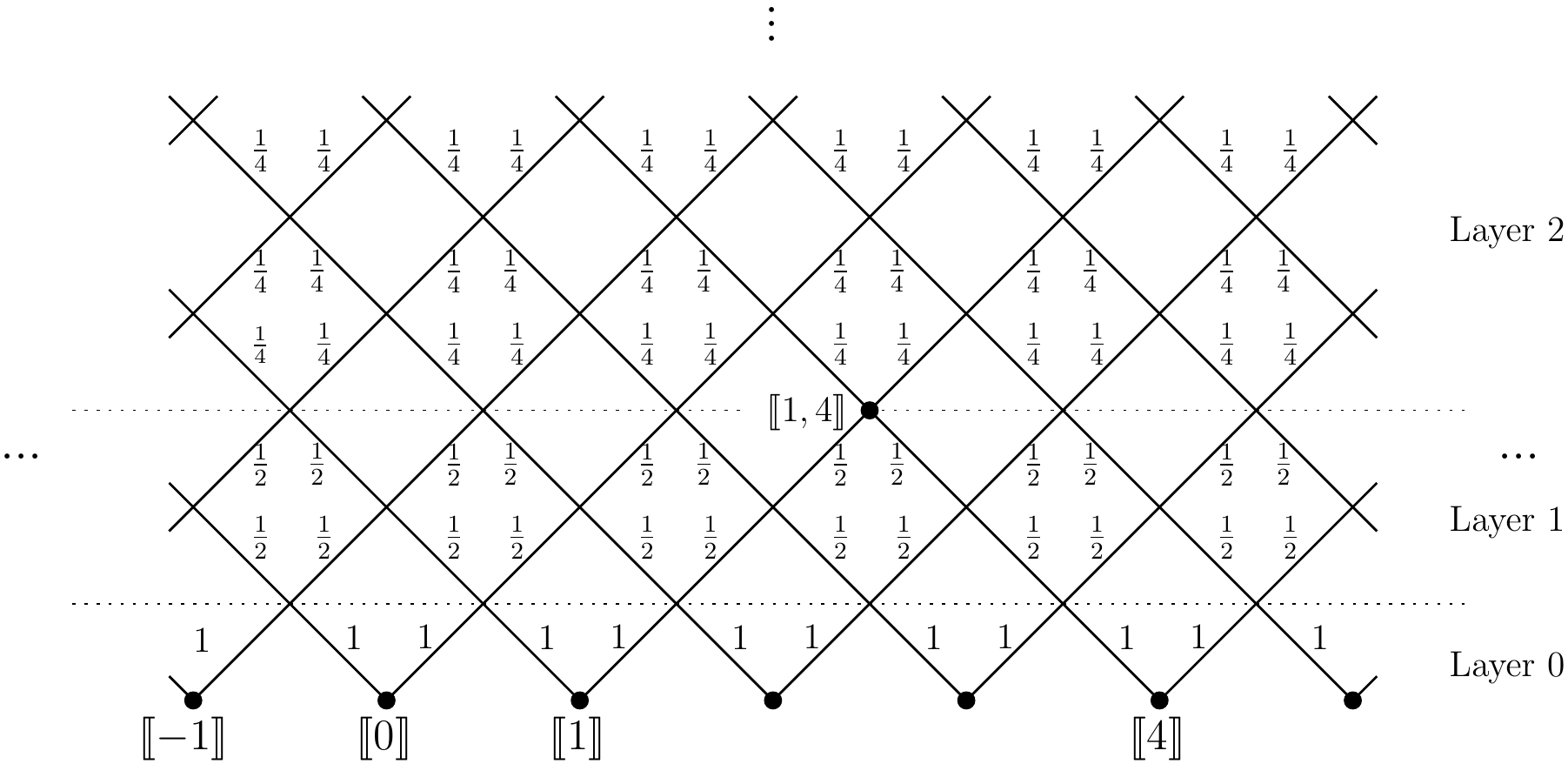}
\end{center}
\caption{The Hasse diagram $\Gamma$ of the poset of discrete intervals. Indicated are the weights assigned by the weight function $w$ as well as the different layers.}\label{fig:gamma}
\end{figure}

We introduce some more terminology and notation. For discrete intervals that are singletons, we use the shorthand $\lb a \rb = \lb a,a \rb$. A \emph{monotone path} in $\Gamma$ is a path $p_1,p_2,\dots,p_k$ where $p_1\subset p_2 \subset \dots \subset p_k$, and if additionally all $p_1,\ldots,p_k$ have either a common left or right endpoint we call the path \emph{straight}.
 For a pair of intervals that share an endpoint, the monotone path between them is unique. The \emph{triangle} of a discrete interval $q=\lb a,b\rb$ is the set of its subintervals; we denote this set by $\Delta(q)$ or $\Delta(\lb a,b\rb)$. The lowest common ancestor of the intervals $p$ and $q$ is denoted by $p \wedge q$. Let $\dist$ denote the shortest path distance in $\Gamma$, i.e., $\dist(a,b):=\inf \{\sum_{pq \in P} w(pq) \;|\; P\text{ is a path from a to b}\}$. For a vertex subset $S\subset V(\Gamma)$, let $\dist(p,S) = \inf_{s\in S} \dist(p,s)$; the notion of distance is extended to sets $A,B$ by letting $\dist(A,B) = \inf_{a\in B, b \in B} \dist(a,b)$.
 Note that if $p$ is a point and $A,B$ are sets we have that $\dist(p,A)+\dist(p,B)\geq \dist(A,B)$ since all distances are non-negative, but the triangle inequality does not hold for sets: It is easy to come up with sets $A,B,C$ such that $\dist(A,B)+\dist(B,C)\leq \dist(A,B)$.
 A \emph{layer} is a maximal subset of edges in $\Gamma$ of the same weight, i.e., layer $i$ has weight $2^{-i}$. See Figure~\ref{fig:gamma} for an illustration of the weights.

For any $x \in \Ints$, the \emph{vertical} at $x/2$ is $V_{x/2} = \{\lb a,b\rb \mid \frac{a+b}{2} = \frac{x}{2}\}$. Note that if $p = \lb a,b \rb$, then $p \in V_{(a+b)/2}$; also, if $x$ is a multiple of $2$, then $\lb x/2 \rb \in V_{x/2}$. A \emph{column} is the set of edges between two consecutive verticals $V_{x/2}$ and $V_{(x+1)/2}$. For any $y \in \mathbb{N}$, the \emph{horizontal} at $y$ is $H_{y} = \{ \lb a,b \rb \mid b-a = y\}$. Note that $H_0 = \{\lb a \rb \mid a \in \Ints\}$. A \emph{row} is the set of edges between two consecutive horizontals $H_y$ and $H_{y+1}$, and the \emph{height} of an edge is the index of the horizontal passing through its lower endpoint. Notice that the weight of edges is weakly decreasing as the height is increasing.

\begin{prop}\label{prop:shortestpaths}
If $p \subseteq q$ are discrete intervals, then any monotone path between them is a shortest path. If $p=\lb a,b\rb$ is a discrete interval and $x \in \Ints$, then the distance from $p$ to the vertical $V_{x/2}$ is realized by the straight monotone path from $p$ to its lowest ancestor in $V_{x/2}$. Finally, if $p$ and $q$ are incomparable, then the union of the straight paths $p\to (p\wedge q)$ and $(p\wedge q) \to q$ is a shortest path from $p$ to $q$.
\end{prop} 

\begin{proof}
For the first claim, if $p\subset q$, then any path between them must contain edges that traverse from a discrete interval of size $|p|$ to a discrete interval of size $|p|+1$, an edge from size $|p|+1$ to $|p|+2$, etc., and an edge from size $|q|-1$ to $|q|$. Any monotone path contains only one of each edge listed. Furthermore, edges where the intervals have identical size have identical weight. Hence, all monotone paths are shortest paths.

\begin{figure}
\begin{center}
\includegraphics[scale=0.65]{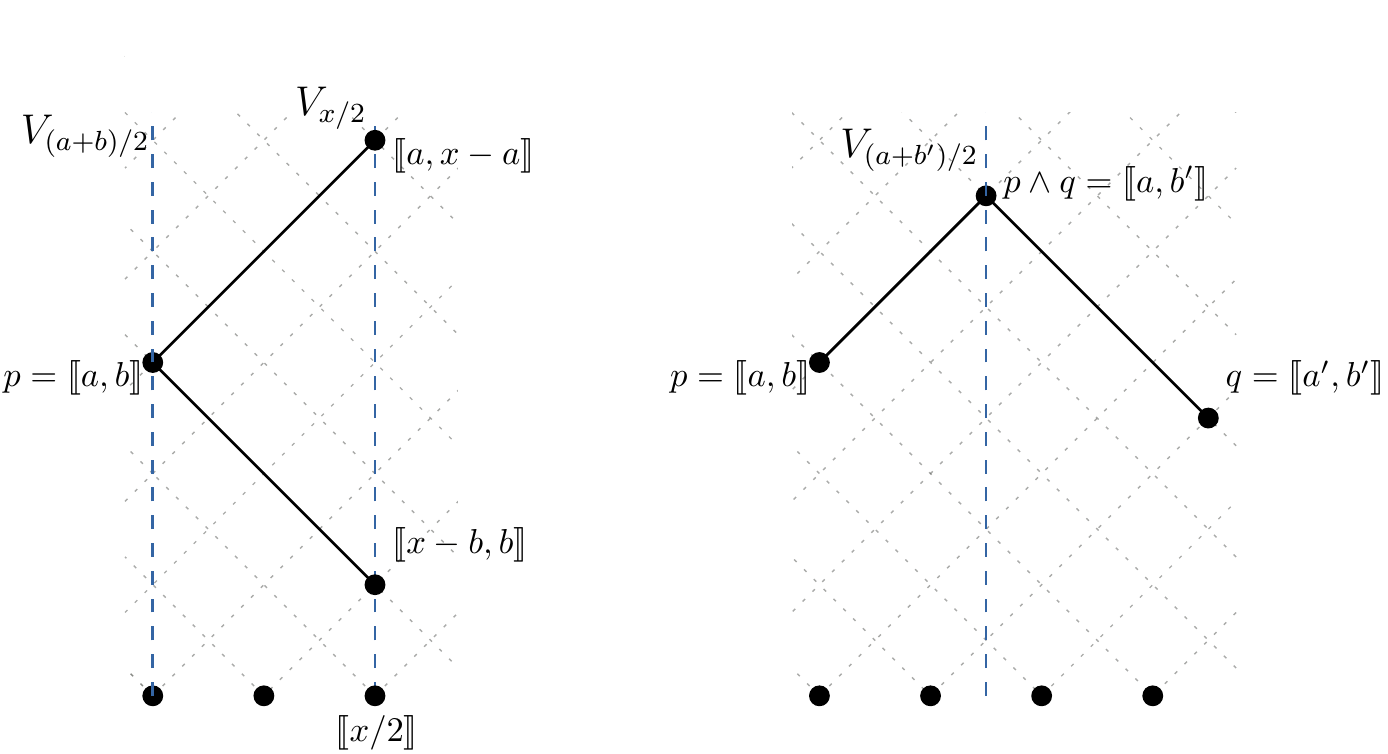}
\end{center}
\caption{Left: Distance from $p$ to a vertical. Right: The distance between a pair of non-comparable points $p$ and $q$.}\label{fig:shortest}
\end{figure}

For the second claim, suppose without loss of generality that $x\geq a+b$. Let $v\in V_{x/2}$ be a vertex where $\dist(p,V_{x/2})=\dist(p,v)$ (see Figure~\ref{fig:shortest}). Since the distance from a point to a set is defined as an infimum, we first show that such a vertex minimizing the distance exists. Note that all but finitely many vertices of $V_{x/2}$ are ancestors of $p$, and among the ancestors the closest one is $\lb a,x-a\rb$ since it minimizes the set of rows that it needs to pass. Therefore, the minimum distance is either realized by $\lb a,x-a\rb$ or a non-ancestor of $p$, of which there are finitely many; consequently, such a vertex $v$ exists, and if it is an ancestor of $p$, then it must be $\lb a,x-a \rb$.

We now use induction on $x-(a+b)$; let $\pi$ denote the straight path from $p$ to $\lb a,x-a\rb$. Clearly if $x-(a+b)\leq 1$ the claim holds. Otherwise suppose (for the sake of contradiction) that there is a vertex $v\in V_{x/2}$ that is not an ancestor of $p$ such that $\dist(p,v)<\dist(p,\lb a,x-a \rb)$; let $\pi'$ be a shortest path from $p$ to $v$. If $\pi'$ intersects $\pi$ at some internal vertex $q=\lb a,b' \rb$, then by induction $\pi'$ is not lengthened if we replace its portion from $q$ to $v$ by the straight path from $q$ to $\lb a,x-a\rb$; but the resulting path $\pi''$ is a path from $p$ to $\lb a,x-a\rb$, which cannot be strictly shorter than $\pi$. If $\pi'$ is internally disjoint from $\pi$, then it must traverse an edge in each column between $V_{(a+b)/2}$ (the vertical containing $p$) and $V_{x/2}$ somewhere below $\pi$. But notice that $\pi$ is not longer than any such path, since it contains exactly one edge from each of these columns, and in each column its edges have larger height, and therefore less or equal weight; this is a contradiction.

% If $x/2>b$, then $V$ is disjoint from $\Delta(p)$, that is, it does not contain any descendants of $p$, therefore the only remaining possible distance minimizing vertex is $\lb a,x-a\rb$. Otherwise, if $v$ is a descendant of $p$, then by the first claim of this lemma, any monotone path is a shortest path, and in particular, the vertex $\lb x-b,b\rb$ minimizes the distance to $p$ among the descendants of $p$ in $V$ (by minimizing the set of rows that it needs to cross). But the unique monotone path from $p$ to $\lb x-b,b\rb$ is at least as long as $\dist(p,\lb a,x-a\rb)$, since both shortest paths have exactly one edge from each column between $V_{(a+b)/2}$ and $V_{x/2}$, but in each column the edge in the path from $p$ to $\lb a,x-a\rb$ is higher and therefore has less or equal weight; this is a contradiction. Therefore, $v=\lb a,x-a\rb$ is a closest neighbor of $p$ in $V$, and the shortest path is a straight monotone path due to the first claim of this lemma.

For the third claim, if $p$ and $q$ are incomparable, then let $\lb a,b\rb=p$ and $\lb a',b'\rb=q$. Without loss of generality, $a < a'$ and $b < b'$. Let $V$ be the vertical passing through $p \wedge q$, that is, let $V=V_{(a+b')/2}$. Note that $\lb a,b\rb \wedge \lb a',b'\rb = \lb a,b'\rb \in V$. Since $V$ separates $p$ and $q$, we have that $\dist(p,q)\geq \dist(p,V)+\dist(V,q)$. By the previous claim, we have $\dist(p,V)=\dist(p,\lb a,((a+b')-a)\rb) = \dist(p,p \wedge q)$. Then, by symmetry, we have $\dist(q,S)=\dist(q,p \wedge q)$,  and the claim holds.
%\todo{maybe add figure?}
\end{proof}

The \emph{left diagonal} at $b$ is $LD_b = \{ \lb x,b\rb \;|\; x \leq b \}$, the \emph{right diagonal} at $a$ is $RD_a = \{ \lb a,x\rb \;|\; x \geq a \}$. Figure~\ref{fig:triangle_lemma} shows examples of both.

\begin{prop}\label{prop:diagonaldistance}
The distance of $LD_b$ and $V_{b+1/2}$ is $1$. The distance of $LD_a$ and $RD_{a+1}$ is $2$. %In particular, any path that is the union of the monotone paths $\lb x,a \rb \rightarrow \lb x,2a+1-x \rb$ and $\lb x,2a+1-x \rb  \rightarrow \lb a+1,2a+1-x \rb$ is a shortest path from $LD_a$ to $RD_{a+1}$ of length $2$; this path contains edges from at most two adjacent layers. 
\end{prop}

\begin{proof}
There is a shortest path from $\lb a,b\rb$ to $V_{b+1/2}$ that is monotone and has $b-a+1$ edges and ends at $\lb a,2b+1-a\rb$ by Proposition~\ref{prop:shortestpaths}. We claim that all of these paths have weight exactly $1$. We use induction on $b-a$. The path from $\lb b\rb$ to $\lb b,b+1\rb$ has a single edge of weight $1$. Consider the shortest path starting at $\lb a,b\rb$.  By induction the shortest path starting at $\lb a+1,b\rb$ has weight $1$, so it is sufficient to show that they have equal weight. Notice that these paths traverse the same horizontal edge rows, except the row from size $|b-a-1|$ to $|b-a|$ that is only traversed by the path of $\lb a+1,b\rb$, and the rows from size $|2b-1-2a|$ to $|2b+1-2a|$ that are only traversed by the path of $\lb a,b\rb$. Notice that the edge row unique to the path of $\lb a+1,b\rb$ is in layer $\lfloor\log(b-a)\rfloor$, while the two edges unique to the path of $\lb a,b\rb$ are both in layer $\lfloor\log(2b-2a)\rfloor = \lfloor\log(2b-2a+1)\rfloor$. Consequently, the edge unique to $\lb a+1,b\rb$ is precisely one layer below the two edges unique to $\lb a,b\rb$, and thus the two paths have equal weight.

To prove the claim about the distance of diagonals, we apply the first claim: $\dist(LD_a,RD_{a+1})\leq \dist(LD_a,V_{a+1/2}) + \dist(V_{a+1/2},RD_{a+1})$ since $V_{a+1/2}$ separates $LD_a$ from $RD_{a+1}$, and both terms on the right hand side are $1$ by the first claim. On the other hand, there is a path of length $2$: the path $\lb a\rb;\lb a,a+1\rb;\lb a+1\rb$.
%The two monotone paths in the third claim have length exactly $1$ by the diagonal-vertical argument, so their union indeed has length $2$ and therefore it is shortest. The previous argument has also shown that the path contains edges from at most two adjacent layers.
\end{proof}

\begin{lemma}\label{lem:triangle}
Let $p$ be a discrete interval and let $\ell \geq 0$.
The weight of any Steiner tree for the terminal set $\{p,\lb 0\rb,\dots,\lb \ell\rb\}$ is at least $2\ell+\dist(\Delta(\lb 0,\ell\rb),p)$.
\end{lemma}

\begin{proof}
Let $p=\lb a,b\rb$ and suppose without loss of generality that $a+b \ge \ell$, that is, $p$ is on or to the right of $V_{\ell/2}$. The proof is by double induction, first on $\ell$, and second for a fixed $\ell$ on the distance $\dist(\Delta(\lb 0,\ell\rb),p)$. Clearly for $\ell=0$, the Steiner tree is at least as long as the distance from $\lb 0\rb$ to $p$. Let $\ell\geq 1$, and let $S$ be the Steiner tree.
We distinguish several cases based on the location of $p$, see Figure~\ref{fig:triangle_lemma}.

\paragraph*{Case 1.} $p \in \Delta(\lb 0,\ell \rb)$, (that is, $a \geq 0$ and $b \leq \ell$)\\
If $p$ has degree $1$ (in $S$), then let $q$ be the nearest vertex within $S$ to $p$ that has degree at least $3$. The tree $S$ is also a Steiner tree for the terminal set $\{q,\lb 0\rb,\dots,\lb \ell\rb\}$, and $\dist(\Delta(\lb 0,\ell\rb),p) = 0 \leq \dist(\Delta(\lb 0,\ell\rb),q)$, so it is sufficient to prove the claim for $q$ instead of $p$. Therefore, without loss of generality, assume that $p$ has degree at least $2$. Let $r$ be a neighbor of $p$. The edge $pr$ defines two subtrees rooted at $p$: one where the shortest path to $p$ traverses $pr$ and one where it does not. Each tree must contain some non-empty sub-interval of the terminals $\lb 0\rb,\dots,\lb \ell\rb$; suppose that $S_1$ contains $\lb 0\rb,\dots,\lb x\rb$ and $S_2$ contains $\lb x+1\rb,\dots,\lb \ell\rb$.

By induction on $\ell$, we have that
\begin{align*}
w(S) &\geq w(S_1) + w(S_2)\\
&\geq 2x + \dist(p,\Delta(\lb 0,x \rb)) + 2(\ell-x-1) + \dist(p,\Delta(\lb x+1,\ell \rb))\\
&\geq 2\ell -2 + \dist(\Delta(\lb 0,x \rb),\Delta(\lb x+1,\ell \rb))\\
&\geq 2\ell -2 + \dist(LD_x,RD_{x+1})\\
&\geq 2\ell,
\end{align*}
where the second inequality follows from the induction hypothesis, the fourth inequality follows as the diagonals separate $\Delta(\lb 0,x \rb)$, and $\Delta(\lb x+1,\ell \rb)$, and the last inequality follows from Proposition~\ref{prop:diagonaldistance}.

\paragraph*{Case 2.} $p \not\in \Delta(\lb 0,\ell \rb)$, (that is, $a < 0$ or $b < \ell$)\\
Without loss of generality, assume that $p$ is a vertex of $V(S)$ that maximizes $\dist(p,\Delta(\lb 0,\ell \rb))$. (Note that for a given tree $S$, the lemma gives the strongest lower bound for such a vertex $p$). Furthermore, among vertices maximizing this distance, there must be at least one vertex $p$ with a neighbor $q$ where $\dist(q,\Delta(\lb 0,\ell \rb))<\dist(p,\Delta(\lb 0,\ell \rb))$. Suppose there is no such vertex $p$; then let $p$ be an arbitrary distance-maximizing vertex. There is a path from $p$ to $\lb 0 \rb$ where $p$ is at positive distance from $\Delta(\lb 0,\ell \rb)$, while $\lb 0 \rb$ is at distance $0$. So there is an edge $p'q'$ on the path where $\dist(q',\Delta(\lb 0,\ell \rb))<\dist(p',\Delta(\lb 0,\ell \rb))=\dist(p,\Delta(\lb 0,\ell \rb))$; this is a contradiction.

So we can suppose without loss of generality that $p$ has a neighbor $q$ such that $\dist(q,\Delta(\lb 0,\ell \rb))<\dist(p,\Delta(\lb 0,\ell \rb))$. If $p$ has degree $1$, then by induction on the distance we have that
\[w(S)=w(pq) + w(S\setminus \{pq\}) \geq w(pq)+  2\ell + \dist(q,\Delta(\lb 0,\ell \rb)) \geq 2\ell+\dist(p,\Delta(\lb 0,\ell\rb)).\]

Suppose now that $p$ has degree at least $2$. Similarly to Case 1, we define the trees $S_1$ containing $\lb 0\rb,\dots,\lb x\rb$ and $S_2$ containing $\lb x+1\rb,\dots,\lb \ell\rb$ based on the branching at $p$.
By induction, we have
\begin{eqnarray}
w(S) &\geq & w(S_1) + w(S_2) \nonumber\\
&\geq& 2x + \dist(p,\Delta(\lb 0,x \rb)) + 2(\ell-x-1) + \dist(p,\Delta(\lb x+1,\ell \rb)) \nonumber\\
&=& 2\ell -2 +  \dist(p,\Delta(\lb 0,x \rb)) + \dist(p,\Delta(\lb x+1,\ell \rb)) \label{eq:f:1}
\end{eqnarray}
It remains to show that $\dist(p,\Delta(\lb 0,x \rb)) + \dist(p,\Delta(\lb x+1,\ell \rb)) \geq 2 + \dist(p,\Delta(\lb 0,\ell \rb))$.

\paragraph*{Case 2a.} $\lb 0,\ell \rb \in \Delta(p)$, (that is, $a \leq 0$ and $b>\ell$)\\
We take shortest paths from $p$ to $\Delta(\lb 0,x \rb)$ and $\Delta(\lb x+1,\ell \rb)$ as suggested by Proposition~\ref{prop:shortestpaths} to lower bound the relevant distances. Note that $\lb 0,\ell \rb \in \Delta(p)$ implies that all vertices in both triangles (and also in $\Delta(\lb 0,\ell\rb)$ are descendants of $p$, so the distance can be realized from $p$ to any of the three triangles is realized by an arbitrary monotone path from $p$ to the tip of the triangle (this path traverses the least amount of horizontal rows). In particular, we can use an arbitrary monotone path $P$ from $p$ to $\lb 0,\ell\rb$ together with the monotone path $P_x$ from $\lb 0,\ell \rb$ to $\lb 0,x \rb$ to realize $\dist(p,\Delta(\lb x+1,\ell \rb))$, and we can use $P$ with the monotone path $P_{x+1}$ from $\lb 0,\ell\rb$ to $\lb x+1,\ell\rb$ to realize $\dist(p,\Delta(\lb x+1,\ell \rb))$. Therefore, 
\begin{eqnarray}
\lefteqn{\dist(p,\Delta(\lb 0,x \rb)) + \dist(p,\Delta(\lb x+1,\ell \rb))} \hspace{4cm} \nonumber\\
 &=& 2w(P) + w(P_x) + w(P_x+1) \nonumber\\
&=& 2\dist(p,\Delta(\lb 0,\ell \rb)) + \dist(\lb 0,x \rb, \lb 0,\ell \rb) + \dist(\lb 0,\ell \rb,\lb x+1,\ell \rb) \nonumber\\
&\geq& 2\dist(p,\Delta(\lb 0,\ell \rb)) + \dist(LD_x, RD_{x+1}) \nonumber\\
&\geq& \dist(p,\Delta(\lb 0,\ell \rb)) + 2, \label{eq:f:2}
\end{eqnarray}

where the last inequality uses Proposition~\ref{prop:diagonaldistance}. Then (\ref{eq:f:1}) and (\ref{eq:f:2}) combined imply that $w(S) \geq 2\ell + \dist(p,\Delta(\lb 0,\ell \rb))$, as claimed.

\begin{figure}
\begin{center}
\includegraphics[scale=0.75]{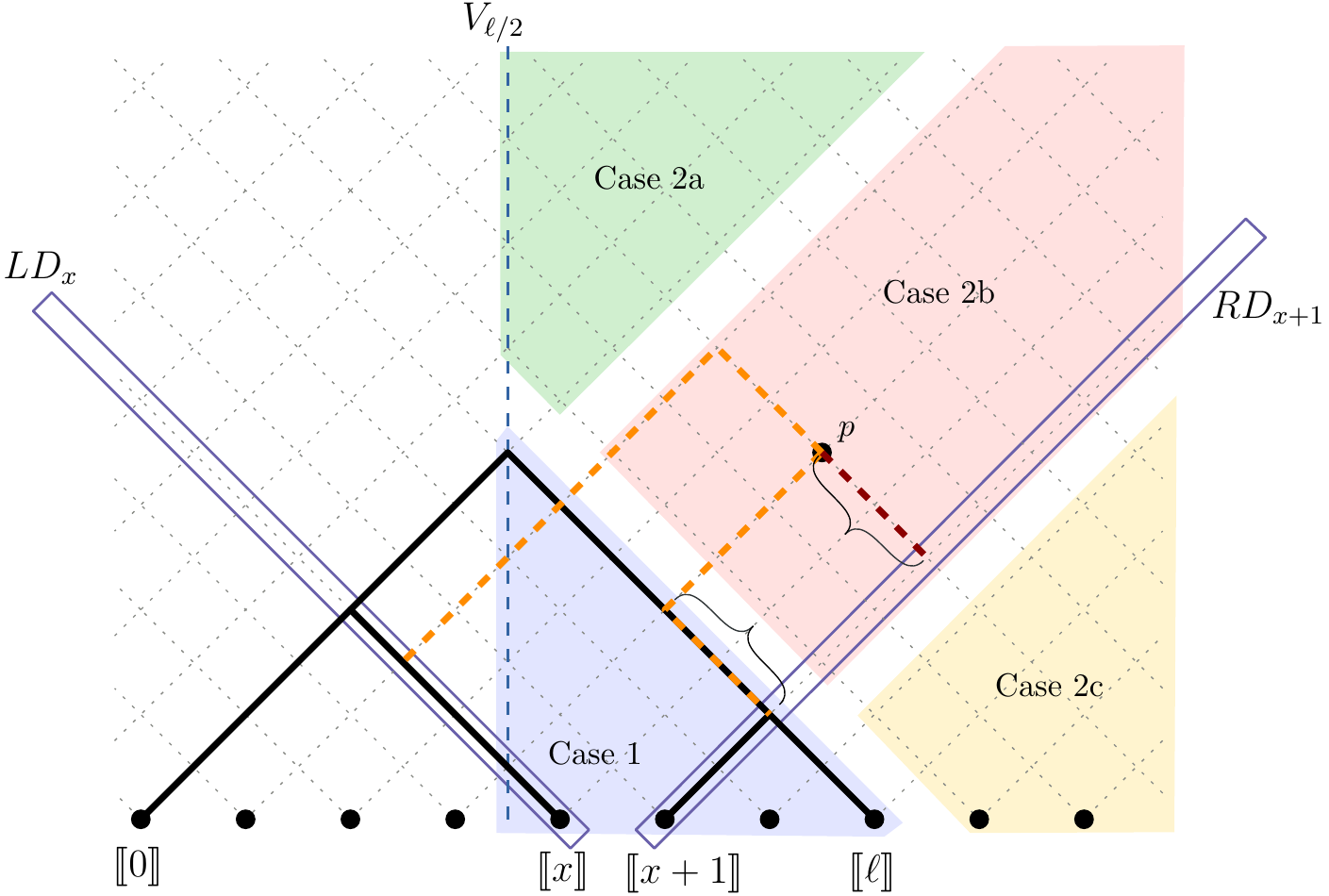}
\end{center}
\caption{The case distinction of Lemma~\ref{lem:triangle} with Case 2b illustrated.}\label{fig:triangle_lemma}
\end{figure}

\paragraph*{Case 2b.} $p$ and $\lb 0,\ell \rb$ are incomparable, and $\lb x+1,\ell \rb \in \Delta(p)$ (that is, $0< a \leq x+1$)\\
As before, it is sufficient to show that $\dist(p,\Delta(\lb 0,x \rb)) + \dist(p,\Delta(\lb x+1,\ell \rb)) \geq 2 + \dist(p,\Delta(\lb 0,\ell \rb))$.
Let $P_1$ be the unique monotone path from $p$ to $\lb a,\ell \rb$ and let $P_2$ be the unique monotone path from $\lb a,\ell \rb$ to $\lb x+1,\ell \rb$. By Proposition~\ref{prop:shortestpaths}, $P_1 \cup P_2$ realizes the distance from $p$ to $\lb x+1,\ell \rb$. Let $P_x$ be a shortest path from $p$ to $\Delta(\lb 0,x \rb)$. Next, we replace $P_2$ by $P_3$ which is defined as the unique monotone path from $p$ to $\lb x+1,b \rb$. Notice that $P_2$ and $P_3$ has the same number of edges, but the edges in $P_3$ are in higher rows and therefore $w(P_3)\leq w(P_2)$. So we have that 
$\dist(p,\Delta(\lb 0,x \rb)) + \dist(p,\Delta(\lb x+1,\ell \rb))= w(P_x) + w(P_1) + w(P_2)\geq w(P_x) + w(P_3) + w(P_1)$. Note that $P_x \cup P_3$ is a path from $LD_x$ to $RD_{x+1}$, so it has length at least $2$, and $P_1$ is a path from $p$ to $\Delta(\lb 0,\ell \rb)$, so its length is at least $\dist(p,\Delta(\lb x+1,\ell \rb))$.

\paragraph*{Case 2c.} $p$ and $\lb x+1,\ell \rb$ are incomparable  (that is, $a > x+1$)\\
We again need that $\dist(p,\Delta(\lb 0,x \rb)) + \dist(p,\Delta(\lb x+1,\ell \rb)) \geq 2 + \dist(p,\Delta(\lb x+1,\ell \rb))$.
The intervals in the area between $LD_x$ and $RD_{x+1}$ are all between $p$ and $\Delta(\lb 0,x \rb)$, so any shortest path from $p$ to $\Delta(\lb 0,x \rb)$ contains a subpath from $LD_x$ to $RD_{x+1}$, and therefore has length at least $2$. The shortest path from $p$ to $\Delta(\lb x+1,\ell \rb)$ is also a path from $p$ to $\Delta(\lb 0,\ell \rb)$ since $\Delta(\lb x+1,\ell \rb) \subset \Delta(\lb 0,\ell \rb)$; therefore, its length is at least $\dist(p,\Delta(\lb 0,\ell \rb))$.
\end{proof}

Lemma~\ref{lem:triangle} argues how a hypothetical Steiner tree must look. For actually constructing a Steiner Tree for an interval of terminals, it is easy to show by induction that there is a tree mimicking a binary tree that contains $2^k-i$ monotone paths traversing layer $i$ for each $i=0,\dots,k-1$ of small weight that gives the following observation. See Figure~\ref{fig:binarytree} for an example.
 
\begin{observation}\label{obs:existssteiner}
	Let $\ell = 2^k-1$ for some positive integer $k$. There is a Steiner tree for the terminal set $\{\lb 0,\ell \rb,\lb 0\rb,\dots,\lb \ell\rb\}$ of weight exactly $2\ell$.
\end{observation}

\begin{figure}
\begin{center}
\includegraphics[scale=0.5]{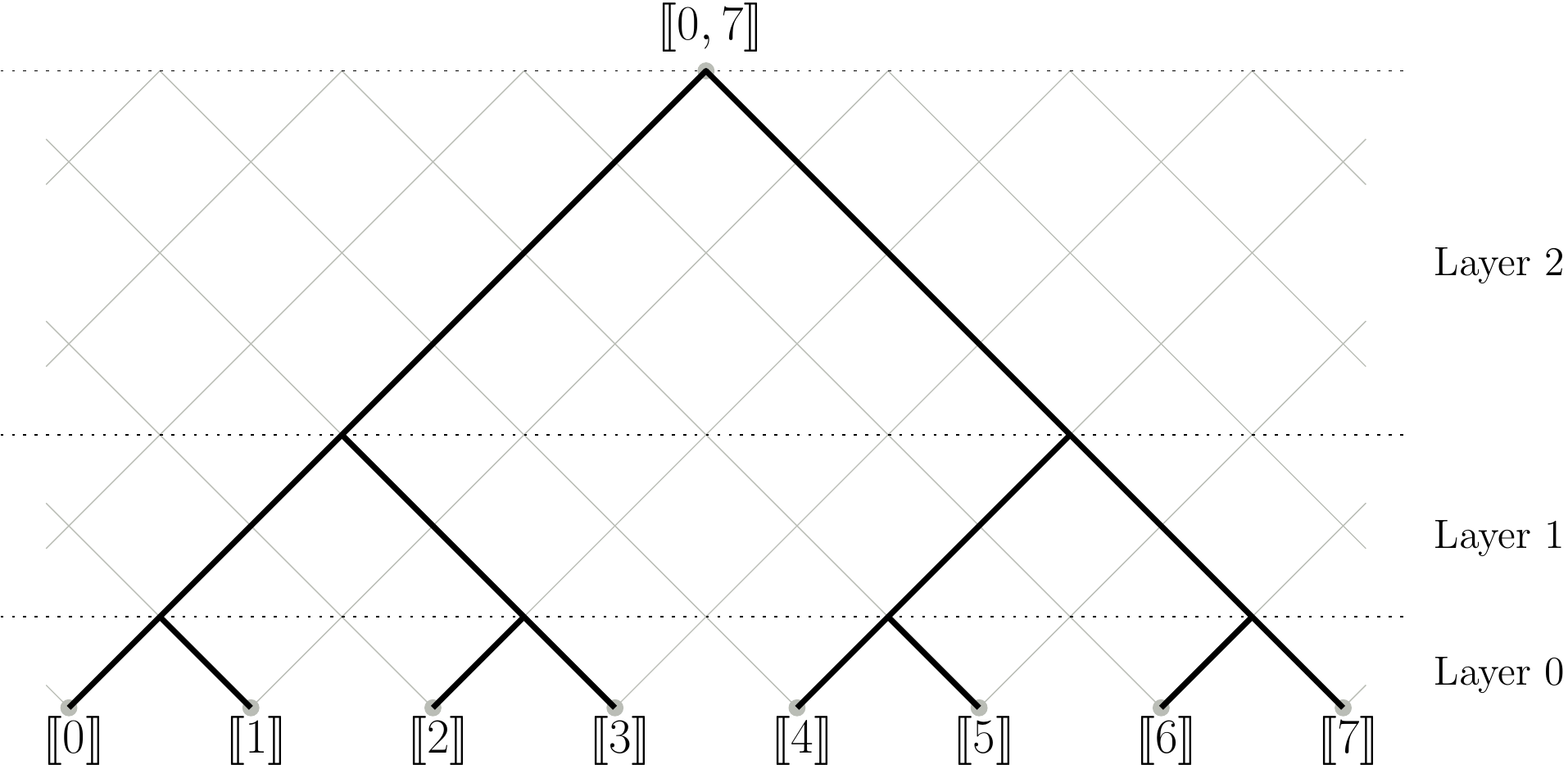}
\end{center}
\caption{Optimal Steiner tree for the terminal set $\{\lb 0,2^3-1 \rb,\lb 0\rb,\dots,\lb 2^3-1\rb\}$ with weight $14$.}\label{fig:binarytree}
\end{figure}

\subsubsection{Construction and properties of the flower gadget} \label{sec:rolled}

The flower gadget is a finite planar graph sharing many properties of $\Gamma$. We give two equivalent definitions. Let $t\geq 4$ be a  power of $2$.

The first definition is to restrict $\Gamma$ to the set $\{\lb a,b\rb \;|\; b-a \leq t/2 -1,\, a \geq 0,\, b \leq t \}$, and identify the vertex pairs $(0,b)$ and $(t,b')$ where $b'=b+t$ for all $b=1,\dots,t$. Let $\Gamma_t$ be the resulting weighted graph. The alternative and more intuitive definition requires the introduction of discrete intervals modulo $t$. Let $a,b\in \{0,1,\dots,t-1\}$. The discrete interval $\lb a,b\rb_t$ is defined as  $\{a,a+1,\dots,b\}$ if $a\leq b$ or as $\{a,\dots,t-1,0,\dots,b\}$ otherwise. Then $\Gamma_t$ is the Hasse diagram for the poset of discrete intervals modulo $t$ of size at most $t/2$, i.e., the Hasse diagram of the set 
\[\big\{\lb a,b\rb_t \,\big|\, a,b \in \{0,1,\dots,t-1\},  (a\leq b \wedge b-a\leq t/2-1) \vee (b<a \wedge t+b-a\leq t/2-1)\big\}.\]
The weighting is identical to $\Gamma$: the weight of an edge $pq$ where $p\subset q=\lb a,b\rb_t$ is $2^{-{\lfloor\log((b-a) \bmod t)\rfloor}}$. Note that $\Gamma_t$ is planar, since it can clearly be drawn on a cylinder, which is topologically equivalent to a punctured plane. See Figure~\ref{fig:flowers} for a planar embedding.

\begin{figure}
\begin{center}
\includegraphics[scale=0.6]{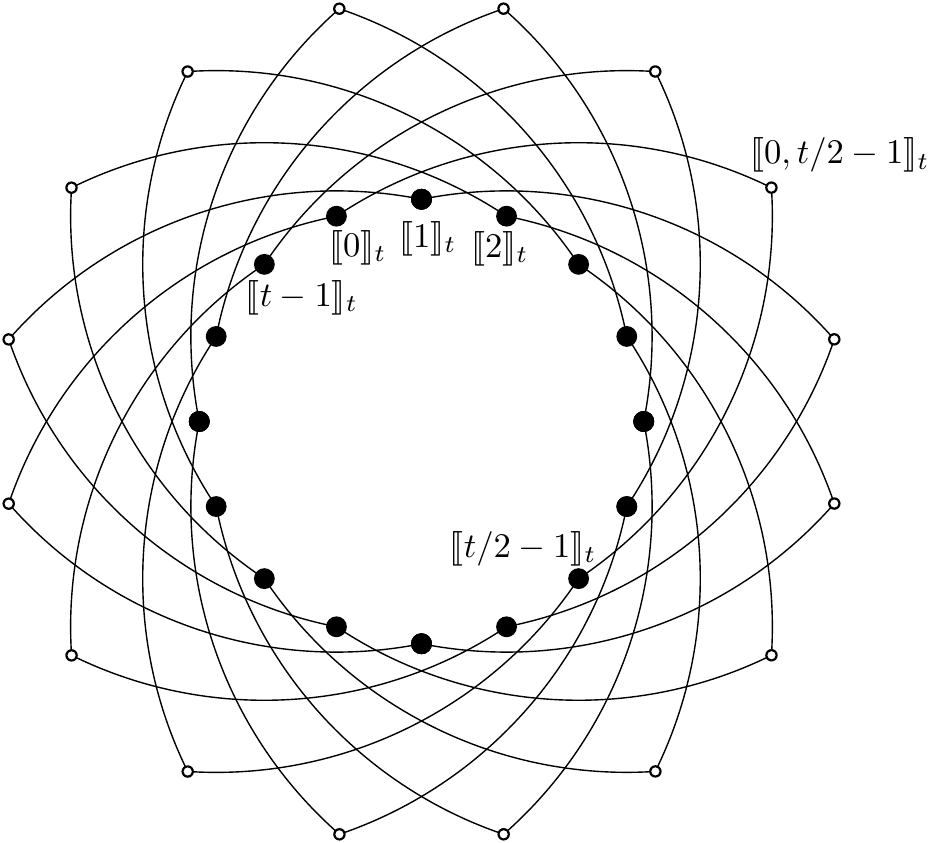}
\hspace*{1cm}
\includegraphics[scale=0.6]{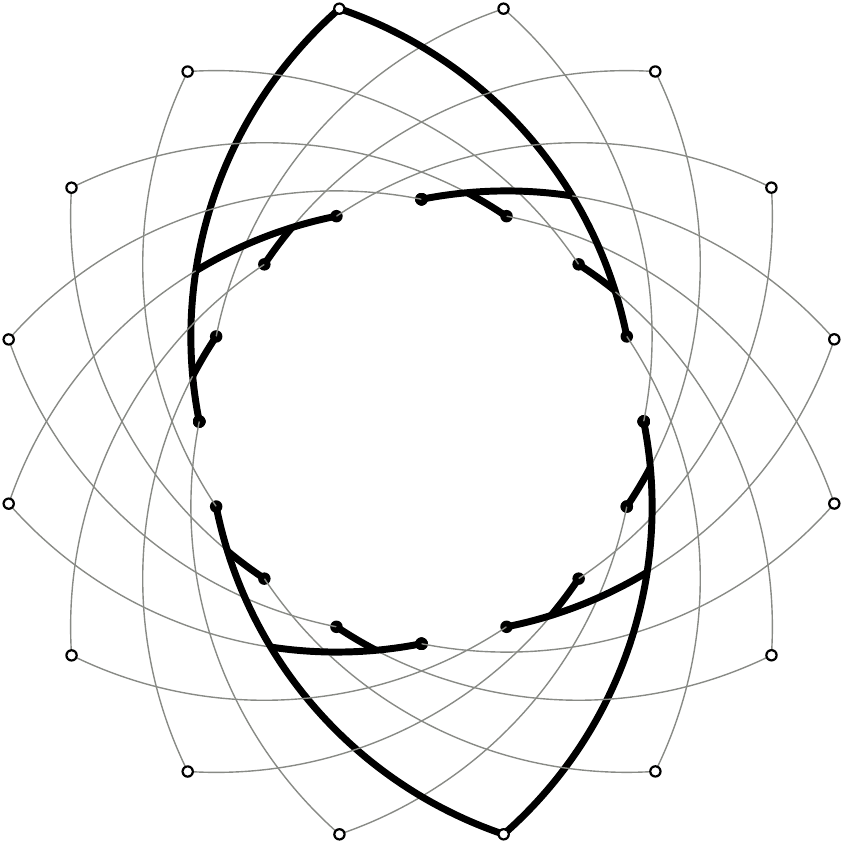}
\end{center}
\caption{Left: The flower gadget, with the terminals as black disks and portals as circles. Right: A canonical solution.}\label{fig:flowers}
\end{figure}

The terminals of the flower gadget are its discrete intervals of size $1$, and the portal vertices are the maximal discrete intervals in the poset. The portals will be used to connect a flower gadget to the rest of the lower bound construction. Note that the portals all reside on the outer face of the embedding (see Figure~\ref{fig:flowers}). We also observe that the terminals of the flower gadget can be covered by a single face, namely the carpel of the flower.

We can define the direction of an edge both in $\Gamma$ and $\Gamma_t$ the following way: $pq$ is a right edge if the right endpoint of $p$ and $q$ are equal, otherwise $pq$ is a left edge.
An edge $e=pq$ in $\Gamma_t$ is isomorphic to an edge $e'=p'q'$ in $\Gamma$ if $|p|=|p'|,|q|=|q'|$, and $pq$ and $p'q'$ are both right or both left edges. It is easy to see that given a subtree $S$ of $\Gamma_t$, there is an isomorphic tree $S'$ in $\Gamma$. For this purpose, we define an isomorphism $\phi:V(S)\rightarrow V(\Gamma)$. Pick an arbitrary vertex $p\in V(S)$, and let $\phi(p):=p'$ where $p'\in V(\Gamma)$ is an arbitrary discrete interval for which $|p|=|p'|$. Using a depth-first search traversal of $S$, the picture of each vertex $v\in V(S)$ is uniquely defined: upon stepping from $v$ to $w$ in $S$, the size of $|w|$ compared to $|v|$ and the direction of the edge $vw$ uniquely identifies $\phi(w)$ (given $\phi(v)$). Note that we do not run into conflicts ($\phi$ is injective), since any cycle in the image would imply the existence of a cycle in $S$, but $S$ is a tree.
For a tree $S$ in $\Gamma_t$, fix an isomorphism $\phi_S$, and let $S_\Gamma$ be the image of the tree. Then we say that the \emph{terminal sequence of $S$} is $\lb a \rb,\dots, \lb b \rb$ if this is the left-to right sequence of vertices in $V(S_\Gamma \cap H_0)$. Consequently, Lemma~\ref{lem:triangle} can also be applied in $\Gamma_t$ in the following sense. For a tree $S$ of $\Gamma'_t$ with terminal sequence $\lb a \rb,\dots, \lb b \rb$ that induces $p$, its weight is at least $2(b-a)+\dist_\Gamma(\Delta(\lb a,b\rb),\phi(p))$.

% Notice that $\phi^{-1}$ extends to the set $\{x\in \Gamma\;|\; |x|\leq t/2\}$, and we can define $\Delta_S$ as $\phi^{-1}(\{x\in \Delta(\lb a,b\rb) \;|\; |x| \leq t/2})$.

The key theorem for using the flower gadget is the following.

\begin{theorem} \label{thm:flower}
Let $S$ be a Steiner forest in the flower gadget (with terminal set $\lb 0 \rb_t, \dots \lb t-1 \rb_t$) where all trees of $S$ contain a portal vertex. Then $S$ has weight at least $2t-4$, and it can have weight exactly $2t-4$ only if it has at least two connected components, each component contains exactly one portal, and for any sequence of $t/2$ consecutive portals, at least one component has its portal there.
Moreover, if $S$ has exactly two components and weight exactly $2t-4$, then $S$ is \emph{canonical}, that is, it induces exactly two opposite portal vertices: $\lb a,\,a+t/2-1\rb_t$ and $\lb a+t/2,\,a-1\rb_t$ for some $a\in \{1,\dots,t/2\}$.
Finally, a canonical Steiner forest $S$ of weight $2t-4$ always exists.
\end{theorem}

\begin{proof}
If $S$ has only one component, then let $p$ be a portal vertex induced by $S$. By Lemma~\ref{lem:triangle}, the tree has weight at least $2t-2$, which is strictly larger than $2t-4$. Hence, $w(S) > 2t-4$ or $S$ has at least two connected components.

Now suppose that there are $k\geq 2$ components; we denote the $i$-th tree of $S$ by $S_i$, and let $p_i$ be a portal vertex from $S_i$. We claim that the set of terminals corresponding to each tree must form a contiguous interval along the terminal face. To see this, suppose for contradiction that $\lb a\rb_t,\lb b \rb_t \in V(S_i)$ and $x \in \lb a,b \rb_t \cap V(S_j)$ for some $i \neq j$. Then the shortest path in $S_i$ from $\lb a\rb_t$ to $\lb b \rb_t$ in the planar embedding  together with an arbitrary planar curve from $\lb a\rb_t$ to $\lb b \rb_t$ inside the terminal face forms a closed curve separating $p_j$ and $\lb x\rb_t$, or has $p_j$ on its boundary and $\lb x\rb_t$ inside. Therefore $S_j$ cannot be disjoint from $S_i$. The claim follows. 

Let $\phi_i=\phi_{S_i}$, and let $\lb a_i\rb \dots \lb b_i\rb\in V(\Gamma)$ be the terminal sequence of $S_i$.
If we apply Lemma~\ref{lem:triangle} for a tree $S_i$, we get 
\[w(S_i) \geq 2(b_i-a_i + \dist(\phi_i(p_i),\Delta(\lb a_i,b_i\rb)).\]
Let $\ell_i = |\lb a_i,b_i\rb|-1$.  
Observe that the terminals are always mapped into $H_0$ and the portal vertex $p_i$ is always mapped into $H_{t/2-1}$ by all $\phi_i$.

Consider the case when the triangle $\Delta(\lb a_i,b_i\rb)$ does not reach $H_{t/2-1}$, so $\ell_i<t/2-1$. Then the distance $\dist(\phi_i(p_i),\Delta(\lb a_i,b_i\rb))$ is at least as big as the distance from the portal set $H_{t/2-1}$ to $H_{\ell_i}$ (Note that $H_{\ell_i}$ passes through $\lb a_i,b_i\rb$.) The weight of the edges below $H_{t/2-1}$ is precisely $4/t$, and the number of edges required on a shortest path from $H_{t/2-1}$ to $H_\ell$ for some $\ell \leq t/2-1$ is $t/2-1-\ell$, so we have that
\begin{equation}\label{eq:horizontals}
\dist(\phi_i(p_i),\Delta(\lb a_i,b_i\rb)) \geq (t/2-1-\ell_i)\frac{4}{t}.
\end{equation}

If the triangle $\Delta(\lb a_i,b_i\rb)$  reaches $H_{t/2}$, that is, if $\ell_i>t/2-1$, then inequality~\eqref{eq:horizontals} still holds because the distance is nonnegative and the right hand side is nonpositive.
Therefore by applying Lemma~\ref{lem:triangle} to each component $S_i$ and then applying the inequality~\eqref{eq:horizontals} we get the following:
\begin{align*}
w(S) &= \sum_{i=1}^k w(S_i)\\
&\geq \sum_{i=1}^k \Big( 2(|\lb a_i,b_i\rb|-1) + \dist(\phi_i(p_i),\Delta(\lb a_i,b_i\rb)) \Big)\\
&\geq 2\sum_{i=1}^k \ell_i + \sum_{i=1}^k (t/2-1-\ell_i)\cdot\frac{4}{t}\\
&= 2(t-k)+ (kt/2-k-(t-k))\cdot\frac{4}{t}\\
&= 2t-4.
\end{align*}
This shows that any Steiner forest of $\Gamma_t$ where each tree contains a portal has weight at least $2t-4$. In case $w(S)=2t-4$, both inequalities in this chain must be equalities, and in particular, each component attains equality for Lemma~\ref{lem:triangle}.

Suppose that $w(S) = 2t-4$ and that a component $S_i$ contains at least two portal vertices, $p_i$ and $p'_i$. If we remove an edge from the tree path that goes from $p_i$ to $p'_i$, then all terminals remain connected to either $p_i$ or $p'_i$. Therefore, we get a Steiner forest that is strictly lighter than $2t-4$ where all components contain a portal; this is contradiction.

Suppose that $w(S)=2t-4$ and there are $t/2$ consecutive portals not induced by any component of $S$; without loss of generality, suppose that these are $\lb 0,t/2-1\rb_t, \lb 1,t/2\rb_t \dots \lb t/2-1, t-2\rb_t$. Let $S_i$ be the tree in $S$ that induces $\lb t/2-1\rb_t$. Then $\phi_i(\lb t/2-1\rb_t) \subset \lb a_i,b_i \rb$. In the second inequality above, equality is only attainable if $\ell_j\leq t/2-1$ for all $j$, since otherwise the distance can be lower bounded by $0$ instead of the negative value we are using. Consequently, $\ell_i\leq t/2-1$. Then in order for the equality to hold for $S_i$, it is necessary that $\phi_i(p_i)$ is at distance $(t/2-1-\ell_i)\frac{4}{t}$ from $\Delta(\lb a_i,b_i \rb)$, which is only possible if $\lb a_i,b_i \rb \subseteq \phi_i(p_i)$ by Lemma~\ref{lem:triangle}. Consequently, $\phi_i(\lb t/2-1\rb_t) \subset \phi_i(p_i)$. Therefore, $p_i \in \{\lb 0,t/2-1\rb_t, \lb 1,t/2\rb_t \dots \lb t/2-1, t-2\rb_t\}$.

Suppose that $w(S) = 2t-4$ and there are exactly two components in $S$. We want to show that $S$ is canonical. If $\ell_1 \neq \ell_2$, then at least one of them is strictly larger than $t/2-1$; suppose that $\ell_1>t/2-1$. Then in the above calculation we can lower bound $\dist(\phi_1(p_1),\Delta(\lb a_1,b_1\rb))$ with $0$ instead of $(t/2-1-\ell_1)\cdot\frac{4}{t}$, which yields a lower bound strictly larger than $2t-4$. Therefore, $\ell_1 = \ell_2=t/2-1$, and the triangles $\phi_1^{-1}(\Delta(\lb a_1,b_1\rb))$ and $\phi_2^{-1}(\Delta(\lb a_2,b_2\rb_t))$ are completely contained in the flower gadget, with their tip being two opposite portal vertices $\lb a_1,b_1\rb_t=\lb a_1,a_1+t/2-1\rb_t$ and $\lb a_2,b_2\rb_t = \lb a_1+t/2,a_1-1\rb_t$. These are the unique portal vertices in $S_1$ and $S_2$ respectively.

Finally, to see that a canonical Steiner forest exists, note that we can simply apply Observation~\ref{obs:existssteiner} to both relevant intervals and observe the complete binary trees will be rooted at two opposite portal vertices.
\end{proof}

\subsection{Verification Gadgets}
We first construct a verification gadget $\gVG$. We use exactly the same gadget as Marx et al. (see Figure~\ref{fig:verification}). \footnote{Figures~\ref{fig:verification} and~\ref{fig:pairverification} from~\cite{MarxPP17} were reproduced here with an author's permission.} The gadget $\gVG$ has $2N+1$ so-called \emph{portals}, which will be identified with or connected to portals of other gadgets. To be precise, the gadget has
\begin{itemize}
\item portals $y[1],\ldots, y[N], w, z[N], \ldots, z[1]$, which appear in this order along the outer face of $\gVG$;
\item vertices $v[i,j]$ for each $i,j \in [N]$;
\item edges from $y[i]$ to $v[1,i]$ of weight $iM_2$ and from $z[i]$ to $v[N,i]$ of weight $iM_3$ for each $i \in [N]$;
\item `horizontal' edges from $v[i,j]$ to $v[i+1,j]$ of weight $M_4$ for each $i \in [N-1]$ and $j \in [N]$;
\item `vertical' edges from $v[i,j]$ to $v[i,j+1]$ of weight $M_3$ for each $i \in [N]$ and $j \in [N] \setminus [i-1]$;
\item edges from $v[i,N]$ to $w$ of weight $M_5-iM_2$ for each $i \in [N]$.
\end{itemize}
We call the edge between $v[i,N]$ to $w$ the $i$-selector of $\gVG$. We actually require the so-called \emph{$S$-reduction} of $\gVG$, denoted $\gVGS$, for a set $S \subseteq [N]$, which is obtained from $\gVG$ by removing the edges from $v[i,N]$ to $w$ for each $i \not\in S$. The following lemma summarizes the properties we require of this gadget.

\begin{figure}
\centering
\includegraphics[scale=0.2]{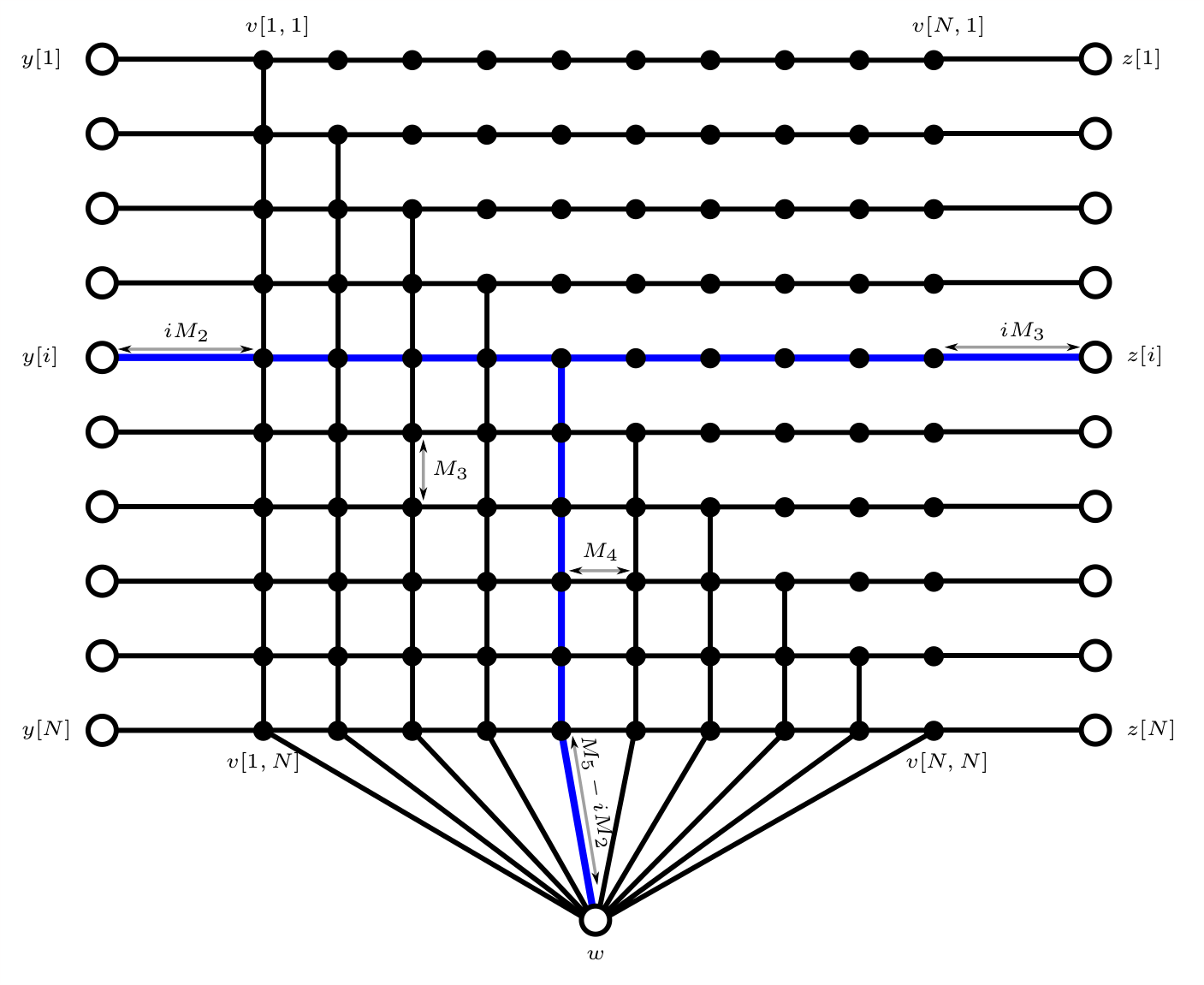}
\caption{The verification gadget $\gVG$ from \cite[Figure 9]{MarxPP17}. The open circles indicate the portals that are connected to other parts of the graph. The blue edges indicate the connected subgraph mentioned in Lemma~\ref{lem:g:veri}(\ref{lem:g:veri:i}).}\label{fig:verification}
\end{figure}

\vspace{-0.1cm}
\begin{lemma} \label{lem:g:veri}
Let $S \subseteq [N]$. Then
\begin{enumerate}[(i)]
\item\label{lem:g:veri:i} for any $i \in [S]$, there is a connected subgraph of $\gVGS$ of weight $M_5+ (N-1)M_4 + NM_3$ that contains $y[i]$, $z[i]$, $w$, and the $i$-selector;
\item\label{lem:g:veri:ii} any connected subgraph $H$ of $\gVGS$ that contains $y[i]$, $z[j]$, and $w$ for $i,j \in [N]$ has weight at least $M_5+ (N-1)M_4 + NM_3$; moreover, if $H$ has weight less than $M_5+ (N-1)M_4 + NM_3 + M_2$, then $i = j$ and $H$ contains the $i$-selector and no other selector edge;
\item\label{lem:g:veri:iii} there is a connected subgraph of $\gVGS$ of weight $(N-1) M_4 + i M_2 + i M_3$ that contains $y[i]$ and $z[i]$ for $i \in [N]$;
\item\label{lem:g:veri:iv} any connected subgraph of $\gVGS$ that contains $y[i]$ and $z[j]$ for $i,j \in [N]$ has weight at least $(N-1) M_4 + i M_2 + i M_3 + 2 \cdot \max\{0,j-i\} \cdot M_3$.
\end{enumerate}
\end{lemma}
\begin{proof}
The first item is proved in~\cite[Lemma~7.5]{MarxPP17}. The second item is proved essentially in~\cite[Lemma~7.6]{MarxPP17}, once we add the observation that every edge in $\gVGS$ is a multiple of $M_2$, and thus having weight less than $M_5+ (N-1)M_4 + NM_3 + M_2$ is equivalent to having weight at most $M_5+ (N-1)M_4 + NM_3$. The third item is immediate from the construction of $\gVG$: follow the edges from $y[i]$ to $v[1,i]$, up to $v[1,j]$, and right to $v[N,j]$ and $z[j]$. It remains to prove the fourth item.

Let $H$ be any connected subgraph of $\gVGS$ that contains $y[i]$ and $z[j]$ for $i,j \in [N]$. Suppose that $H$ has weight less than $(N-1) M_4 + i M_2 + i M_3 + 2\cdot\max\{0,j-i\} \cdot M_3$. Since all edges have positive weight, it suffices to elicit a contradiction under the assumption that $H$ is a path. First note that $H$ has to contain the edges from $y[i]$ to $v[1,i]$ of weight $i M_2$ and from $z[j]$ to $v[N,j]$ of weight $j M_3$. Hence, the remaining edges account for weight at most $(N-1) M_4 + (2 \cdot \max\{0,j-i\} - j) \cdot M_3$. Now consider the possibility that $w$ is contained in $H$. Observe that this means that at least two selector edges are in $H$. However,
$$2(M_5- NM_2) > M_5 > 10NM_4 > (N-1) M_4 + 2N M_3 > (N-1) M_4 + (2\cdot\max\{0,j-i\}-j) \cdot M_3.$$
Hence, $w$ is not contained in $H$. Then $H$ has to contain at least $N-1$ horizontal edges of weight $M_4$ and at least $|j-i|$ vertical edges of weight $M_3$. Therefore, using a case analysis depending on whether $j \leq i$ or $j > i$, $H$ has weight at least $(N-1) M_4 + i M_2 + i M_3 + 2\cdot\max\{0,j-i\} \cdot M_3$, a contradiction to our assumptions. Hence, $H$ has weight at least $(N-1) M_4 + i M_2 + i M_3 + 2\cdot\max\{0,j-i\} \cdot M_3$.
\end{proof}

\begin{figure}
\centering
\includegraphics[width=\textwidth]{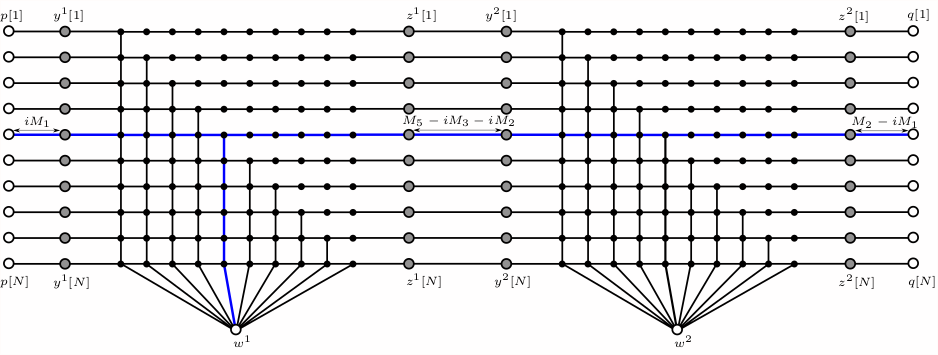}
\caption{The paired verification gadget from \cite[Figure 9]{MarxPP17}. In our gadget $\gLVG$, the edges incident to $w^i$ are omitted depending on $S_i$}\label{fig:pairverification}
\end{figure}

Marx et al. then pair two verification gadgets, see Figure~\ref{fig:pairverification}. We generalize their construction to combine $L$ verification gadgets. To be precise, our gadget $\gLVG$ has
\begin{itemize}
\item portal vertices $p[1],\ldots,p[N], w[1],\ldots,w[L],$ and $q[N],\ldots,q[1]$, which appear in this order along the outer face of $\gLVG$;
\item $L$ verification gadgets $\gVG$, denoted $\gVG^{1},\ldots,\gVG^{\ell}$. Let $y^{\ell}[1],\ldots, y^{\ell}[N], w^{\ell}, z^{\ell}[N], \ldots, z^{\ell}[1]$ denote the portals of $\gVG^{\ell}$ and identify $w[\ell]$ with $w^{\ell}$;
\item edges from $p[i]$ to $y^{1}[i]$ of weight $iM_1$ for $i \in [N]$;
\item edges from $q[i]$ to $z^{L}[i]$ of weight $M_2-iM_1$ for $i \in [N]$;
\item edges $e^\ell_i$ from $z^{\ell}[i]$ to $y^{\ell+1}[i]$ of weight $M_5 - iM_3 - iM_2$ for $i \in [N]$ and $\ell \in [L-1]$, called the connector edges.
\end{itemize}
For $L=2$, this is exactly the paired verification gadget of~\cite{MarxPP17}. We require the \emph{$\mathcal{S}$-reduction} of $\gLVG$, which for $\mathcal{S} = \{S_1,\ldots,S_L\}$ where $S_\ell \subseteq [N]$ for $\ell \in [L]$, contains the $S_\ell$-reduction $\gVGS[S_\ell]$ as the $\ell$-th verification gadget (instead of the plain vanilla $\gVG^{\ell}$) for $\ell \in [L]$.

The following lemma summarizes the properties we require of this gadget.

\begin{lemma} \label{lem:g:veril}
Let $\mathcal{S} = \{S_1,\ldots,S_L\}$, where $S_\ell \subseteq [N]$ for $\ell \in [L]$. Then
\begin{enumerate}[(i)]
\item\label{lem:g:veril:i} for any $\ell \in [L]$ and $i \in S_\ell$, there is a connected subgraph of $\gLVGS$ of weight $L M_5 + L(N-1)M_4 + NM_3 + M_2$ that contains $p[i]$, $q[i]$, $w[\ell]$ and the $i$-selector incident on $w[\ell]$;
\item\label{lem:g:veril:ii} any connected subgraph $H$ of $\gLVGS$ that contains $p[i]$, $q[j]$, and $w[\ell]$ for some $i,j \in [N]$ and $\ell \in [L]$ has weight at least $L M_5 + L(N-1)M_4 + N M_3 + M_2$; moreover, if $H$ has exactly this weight, then $i=j$ and $H$ contains exactly one selector edge, namely the $i$-selector incident on $w[\ell]$.
\end{enumerate}
\end{lemma}
\begin{proof}
The first item is essentially proved in~\cite[Lemma~7.7]{MarxPP17}, but only for $L=2$. The extension to arbitrary $L$ is straightforward and described here only for sake of completeness. For each $l \in [L] \setminus\{\ell\}$, it follows from Lemma~\ref{lem:g:veri}(\ref{lem:g:veri:iii}) that $\gVGS[S_{l}]$ contains a connected subgraph of weight $(N-1)M_4 + iM_2 + iM_3$ that contains $y^{l}[i]$ and $z^{l}[i]$. It follows from Lemma~\ref{lem:g:veri}(\ref{lem:g:veri:i}) that $\gVGS[S_{\ell}]$ contains a connected subgraph of weight $M_5 + (N-1)M_4 + NM_3$ that contains $y^{\ell}[i]$, $z^{\ell}[i]$, $w^{\ell} = w[\ell]$, and the $i$-selector incident on $w^{\ell} = w[\ell]$. To connect these subgraphs, add the edges between $p[i]$ and $y^{1}[i]$ of weight $iM_1$, between $z^{L}[i]$ and $q[i]$ of weight $M_2 - iM_1$, and for $l \in [L-1]$ between $z^{l}[i]$ and $y^{l+1}[i]$ of weight $M_5 - iM_3-iM_2$ each. Note that this indeed yields a connected subgraph of $\gLVGS$ that contains $p[i]$, $q[i]$, $w[\ell]$ and the $i$-selector incident on $w[\ell]$. The total weight is
\begin{eqnarray*}
\lefteqn{(L-1) ((N-1)M_4 + iM_2 + iM_3) + M_5 + (N-1)M_4 + NM_3 +} \hspace{5cm} \\
\lefteqn{+ iM_1 + M_2 - iM_1 + (L-1) (M_5 - iM_3-iM_2)} \hspace{4.5cm} \\
& = & LM_5 + L(N-1)M_4 + NM_3 + M_2,
\end{eqnarray*}
as claimed.

For the second item, let $H$ be any connected subgraph of $\gLVGS$ that contains $p[i]$, $q[j]$, and $w[\ell]$ for some $i,j \in [N]$ and $\ell \in [L]$. Suppose that $H$ has weight less than $L M_5 + L(N-1)M_4 + (N-1) M_3 + M_2$.

Note that $H$ contains at least $L-1$ connector edges of weight at least $M_5 - NM_3 - NM_2$ each. Moreover, the selector edge incident on $w[\ell]$ contributes at least $M_5 - NM_2$ to the total weight. Suppose that $H$ contains at least $c$ more connector or selector edges. Then $H$ has total weight at least
\begin{eqnarray*}
\lefteqn{(L+c) M_5 - (L+c)NM_3 - (L+c+1)NM_2} \hspace{5cm} \\
& = & (L+1) M_5 - (L+1)NM_3 - (L+2)NM_2 + \\
&& (c-1) \cdot (M_5 - NM_3 - NM_2) \\
& > & (L+1) M_5 - (L+1)NM_3 - (L+2)NM_2 \\
& > & (L+1) M_5- 2LNM_3-2LNM_2 \\
& > & (L+1)M_5 - 4M_4 \\
& > & LM_5 + 6LNM_4 \\
& > & LM_5 + L(N-1)M_4 + (N-1) M_3 + M_2,
\end{eqnarray*}
using that $M_i > 10NLM_{i-1}$. Hence, $H$ contains exactly $L-1$ connector edges, say $e^{l}_{i_l}$ for $i_l \in [N]$ and $l \in [L-1]$, and exactly one selector edge, namely one incident on $w[\ell]$. Let $i_0 = i$ and $i_L = j$.

Consider the $\ell$-th verification gadget. In this gadget, $H$ contains $y^{\ell}[i_\ell]$, $z^{\ell}[i_{\ell+1}]$, and $w[\ell] = w^{\ell}$. By Lemma~\ref{lem:g:veri}(\ref{lem:g:veri:ii}), $H$ has weight at least $M_5 + (N-1)M_4 + NM_3$ in this gadget. In each of the other gadgets, $H$ contains $y^{l}[i_l]$ and $z^{l}[i_{l+1}]$; hence, by Lemma~\ref{lem:g:veri}(\ref{lem:g:veri:iv}), $H$ has weight at least $(N-1)M_4 + i_lM_2 + i_{l}M_3 + 2\cdot\max\{0,i_{l+1}-i_l\} \cdot M_3$ in these gadgets. Moreover, $H$ has to contain the edges from $p[i]$ to $y^1[i]$ of weight $iM_1 = i_0M_1$ and from $q[j]$ to $z^L[j]$ of weight $M_2-jM_1 = M_2-i_LM_1$. Then the total weight of $H$ is at least:
\begin{eqnarray*}
\lefteqn{\left(\sum_{l \in (\{0\} \cup [L-1]) \setminus\{\ell-1\}} (N-1)M_4+i_lM_2 + i_{l}M_3 + 2\cdot\max\{0,i_{l+1}-i_l\} \cdot M_3 \right) +}\hspace{2cm}\\
\lefteqn{+ M_5 + (N-1)M_4 + (N-1)M_3 + \left(\sum_{l \in [L-1]} M_5 - i_lM_3 - i_lM_2 \right)+}\hspace{2cm}\\
\lefteqn{+ i_0M_1 + M_2 - i_LM_1}\hspace{2cm}\\
& = & LM_5 + L(N-1)M_4 + N M_3 + i_0M_1 + M_2 - i_LM_1 + \\
&& + \left(\sum_{l \in (\{0\} \cup [L-1]) \setminus\{\ell-1\}} i_lM_2 + i_{l}M_3 + 2\cdot\max\{0,i_{l+1}-i_l\} \cdot M_3\right) \\
&& + \left(\sum_{l \in [L-1]} - i_lM_3 - i_lM_2 \right)\\
& = & LM_5 + L(N-1)M_4 + N M_3 + i_0M_1 + M_2 - i_LM_1 + \\
&& + \left(\sum_{l=0}^{\ell-1} i_lM_2 + i_{l}M_3 + 2\cdot\max\{0,i_{l+1}-i_l\} \cdot M_3 - i_{l+1}M_3-i_{l+1}M_2\right)\\
&& + \left(\sum_{l=\ell}^{L-1} i_lM_2 + i_{l}M_3 + 2\cdot\max\{0,i_{l+1}-i_l\} \cdot M_3 - i_{l}M_3-i_{l}M_2\right)\\  
%& = & LM_5 + L(N-1)M_4 + (N-1) M_3 + i_0M_1 + M_2 - i_LM_1 + \\
%&& + (i_0-i_\ell)M_2 + (i_0-i_\ell)M_3 + \left(\sum_{l=0}^{\ell-1}2\cdot\max\{0,i_{l+1}-i_l\} \cdot M_3\right)\\
%&& + \left(\sum_{l=\ell}^{L-1} 2\cdot\max\{0,i_{l+1}-i_l\}\cdot M_3\right)
\end{eqnarray*}
Consider the first summation and let $l \in \{0,\ldots,\ell-1\}$. If $i_l > i_{l+1}$, then the $l$-th summand contributes at least $M_2 + M_3 > NM_1$ to the total. If $i_l < i_{l+1}$, then the $l$-th summand contributes at least $M_3-M_2 > NM_1$ to the total. Otherwise, the $l$-th summand contributes $0$ to the total. Consider the second summation and let $l \in \{\ell,\ldots,L_1\}$. Each summand can only contribute a non-negative number to the total, and if $i_l < i_{l+1}$, then it contributes at least $2M_3 > NM_1$. Since
$$LM_5 + L(N-1)M_4 + N M_3 + i_0M_1 + M_2 - i_LM_1 + NM_1 > LM_5 + L(N-1)M_4 + (N-1) M_3 + M_2,$$
it follows that $i_l = i_{l+1}$ for $l \in \{0,\ldots,\ell-1\}$ and that $i_l \geq i_{l+1}$ for $l \in \{\ell,\ldots,L-1\}$. In particular, the contribution of both summations to the total is $0$. Hence, the total weight of $H$ is at least $LM_5 + L(N-1)M_4 + N M_3 + i_0M_1 + M_2 - i_LM_1$. It follows that the part of $H$ in the $\ell$-th verification gadget contributes less than
$$M_5 + (N-1)M_4 + NM_3 + NM_1 < M_5 + (N-1)M_4 + NM_3 + M_2,$$
which implies that $i_\ell = i_{\ell+1}$ by Lemma~\ref{lem:g:veri}(\ref{lem:g:veri:ii}). Hence, $i_L \leq i_0$ and in fact, $i_L = i_0$ and thus, $i_0 = \cdots = i_L$. Therefore, the total weight of $H$ is at least $L M_5 + L(N-1)M_4 + N M_3 + M_2$, a contradiction. Hence, the total weight of $H$ is at least $L M_5 + L(N-1)M_4 + N M_3 + M_2$.

The above arguments imply that if $H$ has weight exactly $L M_5 + L(N-1)M_4 + N M_3 + M_2$, then $H$ contains exactly $L-1$ connector edges, say $e^{l}_{i_l}$ for $i_l \in [N]$ and $l \in [L-1]$, and exactly one selector edge, namely one incident on $w[\ell]$. Moreover, $i_0 = \cdots = i_L$ and thus $i=j$.
\end{proof}

\subsection{Construction}\label{subsec:constr}
In the next two subsections, we aim to prove the following theorem, which will quickly imply Theorem~\ref{thm:g:bound}.

\begin{theorem} \label{thm:g:weighted}
Let $\mathcal{M}$ be an instance of {\pGT}, with associated integers $n$ and $k$. Then in time polynomial in $n$ and $k$, one can construct an integer $K_{\mathcal{M}}$ and a planar graph $G_{\mathcal{M}}$ with positive edge weights and set $T_{\mathcal{M}}$ of terminals such that
\begin{itemize}
\item $G_{\mathcal{M}}$ has size $O(k^2 n^5)$;
\item $T_{\mathcal{M}}$ can be covered by $k(k-1)+1$ faces of $G_{\mathcal{M}}$;
\item each edge has weight $O(k^{14}n^{22})$;
\item $\mathcal{M}$ admits a solution if and only if $G_{\mathcal{M}}$ admits a Steiner tree of weight at most $K_{\mathcal{M}}$.
\end{itemize}
\end{theorem}

Consider an instance $\mathcal{M}$ of {\pGT} consisting of two integers $n$ and $k$, and $k^2$ sets $M_{a,b} \subseteq [n] \times [n]$ for $a,b \in [k]$. By increasing $n$ if necessary, we may assume that $n$ is a power of~$2$. Throughout, let $N = n^2$, $L=n$, $t = 2L = 2n$, and $M = 10k^2NL$; observe that $M > 10NL$ as required. Figure~\ref{fig:bigpic} is provided to get a better understanding of the construction.

\begin{figure}
\centering
\def\svgwidth{\textwidth}
\includegraphics[scale=0.8]{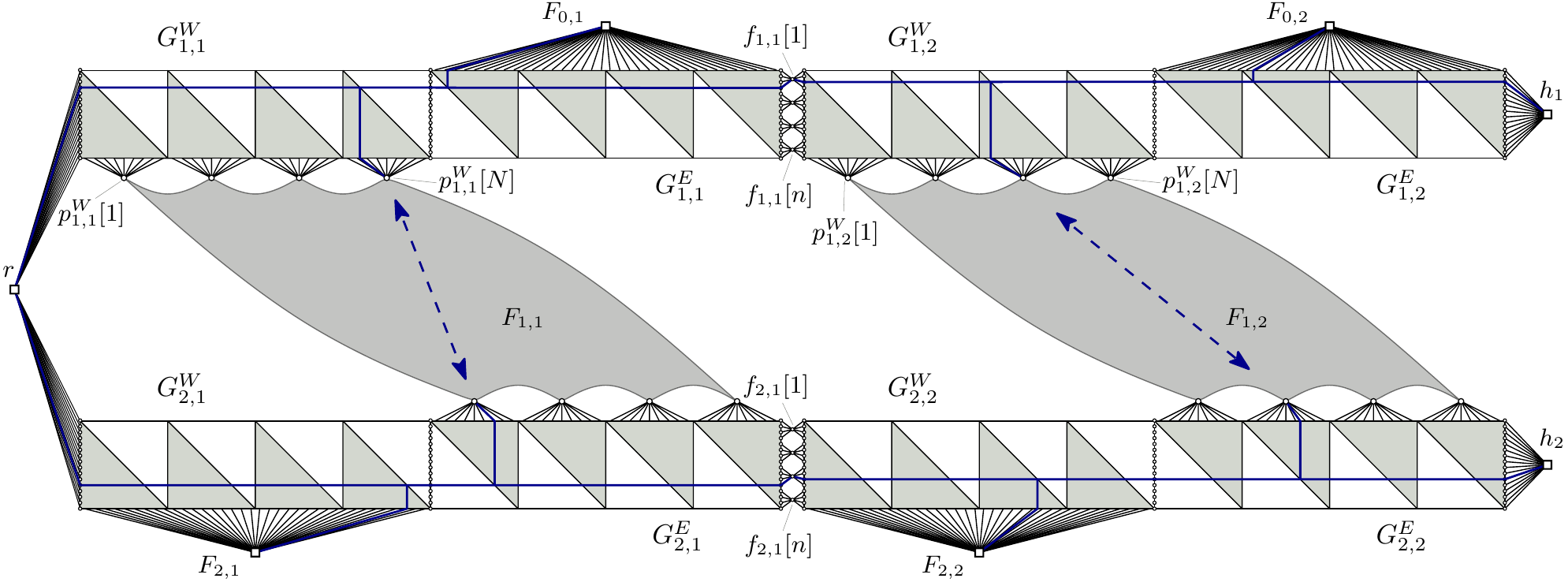}
\caption{The graph $G_{\mathcal{M}}$. In both rows, two copies of $\gLVG$ are fused as described in Subsection~\ref{subsec:constr}. The gray clouds that connect the copies of $\gLVG$ vertically indicate flower gadget, and the blue arrows indicates the matched entry points into the flower gadget.}\label{fig:bigpic}
\end{figure}

For each $a,b \in [k]$, we create two gadgets $\gwest = \gLVGS[\swest]$ and $\geast = \gLVGS[\seast]$ for well-chosen sets $\swest$ and $\seast$. The $w$-portals of $\gwest$ will face south, while the $w$-portals of $\geast$ will face north (i.e.~$\geast$ is $\gLVGS[\seast]$ rotated by 180 degrees). The $q$-portals of both gadgets will then be connected: we place an edge from $q^W[j]$ to $q^E[N-j+1]$ for each $j \in [N]$. The idea will be that selecting a `row' in $\gwest$ and $\geast$ corresponds to selecting a valid choice $(x_a,y_b) \in M_{a,b}$. We then add a flower gadget between $\swest$ and $\seast[a,b+1]$ to ensure that the same $y_b$ is chosen in each column of the {\pGT} instance, while a simpler construction ensures that the same $x_a$ is chosen in $\swest[a+1,b]$. We now describe the construction in more detail.

Let $a,b\in [k]$. We aim to construct $\swest$ and $\seast$, so that $\gwest$ and $\geast$ are well defined. For each $l \in [L]$, let $\sswest{l} = \{(i-1)n+l \mid (i,l) \in M_{a,b}\}$ and let $\sseast{l} = \{N - ((i-1)n+l) + 1 \mid (i,l) \in M_{a,b}\}$. Then $\swest = \{\sswest{1},\ldots,\sswest{L}\}$ and $\seast = \{\sseast{1},\ldots,\sseast{L}\}$. Now let $\gwest = \gLVGS[\swest]$ and $\geast = \gLVGS[\seast]$. We use $\pwest{1},\ldots,\pwest{N},\wwest{1},\ldots,\wwest{L},\qwest{N},\ldots,\qwest{1}$ to denote the portals of $\gwest$ and $\peast{1},\ldots,\peast{N},\weast{1},\ldots,\weast{L},\qeast{N},\ldots,\qeast{1}$ to denote the portals of $\geast$. Now we connect $\qwest{j}$ with $\qeast{N-j+1}$ for each $j \in [N]$ by a join edge $\ej$ of weight $M_6$. Denote the resulting gadget by $\Gab$.

We now fuse the gadgets $\Gab$ for fixed $a \in [k]$. Let $b \in [k-1]$ and $i \in [n]$. Create a new vertex $\fabi{i}$. For each $l \in [n]$, add an edge from $\peast{N-((i-1)n+l)+1}$ to $\fabi{i}$ and from $\fabi{i}$ to $\pwest[a,b+1]{(i-1)n + l}$, both of weight $M_6$. The idea of $\fabi{i}$ is that it allows us to switch the value selected in our solution of the {\pGT} instance between column $b$ and column $b+1$, while the value selected for row $a$ remains the same (namely $i$).

Let $b \in [k]$. For each $a \in [k-1]$, create a flower gadget $\fab$ of size $t$. Since $n$ is a power of~$2$, this is indeed possible. Multiply all weights in the gadget by $2^{\log t} M_7 = t M_7$, so that each weight is at least $M_7$ and at most $O(LM_7)$. Identify the vertex $\lb l,\,l+t/2-1\rb_t$ of the flower gadget with $\wwest{l}$ and the vertex $\lb l+t/2,\,l-1\rb_t$ of the flower gadget with $\weast[a+1,b]{l}$. If the Steiner tree in $\fab$ is canonical, then we can ensure that the value selected in our solution of the {\pGT} instance for column $b$ is the same, namely $l$, throughout. Finally, create a single terminal vertex $\fab[0,b]$ and identify it with $\weast{l}$ for all $l \in [L]$, and create a single terminal vertex $\fab[k,b]$ and identify it with $\wwest{l}$ for all $l \in [L]$. We call these the dummy terminals.

As a last step, create a terminal $r$ and $k$ terminals $h_1,\ldots,h_k$. Add an edge of weight $M_6$ from $r$ to all vertices $\pwest[a,1]{j}$ for $a \in [k]$ and $j \in [N]$. For each $a \in [k]$ and all $j \in N$, add an edge of weight $M_6$ from $\qeast[a,k]{j}$ to $h_a$. 
For notational convenience, we will sometimes write that $\fabi[a,0]{i} = r$ and $\fabi[a,k]{i} = h_a$ for $a \in [k]$ and $i \in [n]$.

Finally, let $K_{\mathcal{M}} := k(k-1) \cdot (2t-4) \cdot tM_7 + 3k^2 M_6 + 2k^2 \cdot (L M_5 + L(N-1)M_4 + N M_3 + M_2)$.

This completes the construction. Observe that the resulting graph $G_{\mathcal{M}}$ is planar. Moreover, $G_{\mathcal{M}}$ has exactly $k(k-1) + 1$ faces that jointly contain all terminals: $k(k-1)$ faces that form the carpels of the flower gadgets, plus the outer face of $G_{\mathcal{M}}$. Finally, observe that $G_{\mathcal{M}}$ has $O(k^2 N^2 L + k^2 L^2) = O(k^2 n^5)$ vertices.

\subsection{Correctness}

\begin{lemma}
If $\mathcal{M}$ admits a solution, then $G_{\mathcal{M}}$ admits a Steiner tree of weight at most
$$K_{\mathcal{M}} = k(k-1) \cdot (2t-4) \cdot tM_7 + 3k^2 M_6 + 2k^2 \cdot (L M_5 + L(N-1)M_4 + N M_3 + M_2).$$
\end{lemma}
\begin{proof}
Let $x_1,\ldots,x_k,y_1,\ldots,y_k$ be a solution. We construct a tree as follows. For each $a,b \in [k]$, it follows from Lemma~\ref{lem:g:veril}(\ref{lem:g:veril:i}) and the construction of $\swest$ that $\gwest$ has a connected subgraph of weight $L M_5 + L(N-1)M_4 + N M_3 + M_2$ that contains $\pwest{x_a \cdot n + y_b}$, $\qwest{x_a \cdot n + y_b}$, and $\wwest{y_b}$. Similarly, $\geast$ has a connected subgraph of weight $L M_5 + L(N-1)M_4 + N M_3 + M_2$ that contains $\peast{(n-x_a+1) \cdot n + (n-y_b+1)}$, $\qeast{(n-x_a+1) \cdot n + (n-y_b+1)}$, and $\weast{y_b}$. Since $\qwest{x_a \cdot n + y_b}$ and $\peast{(n-x_a+1) \cdot n + (n-y_b+1)}$ are connected by a join edge, we obtain a connected subgraph $\Hab$ that contains $\pwest{x_a \cdot n + y_b}$ and $\peast{(n-x_a+1) \cdot n + (n-y_b+1)}$ of total weight $M_6 + 2(L M_5 + L(N-1)M_4 + N M_3 + M_2)$. Hence, the total weight of the union of the connected subgraphs $\Hab$ over all $a,b \in [k]$ is $k^2 M_6 + 2k^2 \cdot (L M_5 + L(N-1)M_4 + N M_3 + M_2)$.

Let $b \in [k]$. For each $a \in [k-1]$, it follows from Theorem~\ref{thm:flower} that $\fab$ has a canonical Steiner forest $\HFab$ on connector vertices $\wwest{y_b} = \lb y_b,\,y_b+t/2-1\rb_t$ and $\weast[a+1,b]{y_b} = \lb y_b+t/2,\,y_b-1\rb_t$ of total weight $2t-3 \cdot tM_7$. Hence, the total weight of the Steiner forests $\HFab$ over all $a \in [k-1],b \in [k]$ is $k(k-1) \cdot (2t-4) \cdot tM_7$. Observe that $\HFab$ has two connected components: the first is attached to $\Hab$ through $\wwest{y_b} = \lb y_b,\,y_b+t/2-1\rb_t$; the second is attached to $\Hab[a+1,b]$ through $\weast[a+1,b]{y_b} = \lb y_b+t/2,\,y_b-1\rb_t$.

Finally, for each $a \in [k], b \in [k-1]$, we select $\fabi{x_a}$. Note that $\fabi{x_a}$ is adjacent to $\peast{(n-x_a+1)\cdot n+(n-y_b+1)}$, and thus to $\Hab$, through an edge of weight $M_6$. Similarly, $\fabi{x_a}$ is adjacent to $\pwest[a,b+1]{x_a\cdot n + y_b}$, and thus to $\Hab[a,b+1]$, through an edge of weight $M_6$. Moreover, for each $a \in [k]$, $r$ is adjacent to $\pwest[a,1]{x_a \cdot n + y_1}$, and thus to $\Hab[a,1]$, through an edge of weight $M_6$. Similarly, $h_a$ is adjacent to $\peast[a,k]{(n-x_a+1) \cdot n + (n-y_k+1)}$, and thus to $\Hab[a,k]$, through an edge of weight $M_6$. Let $H$ denote the resulting subgraph.

Observe that $H$ is connected by construction. Moreover, the weight of $H$ is
$$k(k-1) \cdot (2t-4) \cdot tM_7 + 3k^2 M_6 + 2k^2 \cdot (L M_5 + L(N-1)M_4 + N M_3 + M_2).$$
Furthermore, $H$ contains all terminals, including the dummy terminals. Hence, by taking a spanning tree of $H$, the lemma follows.
\end{proof}

\begin{lemma}
If $G_{\mathcal{M}}$ admits a Steiner tree of weight at most
$$K_{\mathcal{M}} = k(k-1) \cdot (2t-4) \cdot tM_7 + 3k^2 M_6 + 2k^2 \cdot (L M_5 + L(N-1)M_4 + N M_3 + M_2),$$
then $\mathcal{M}$ admits a solution.
\end{lemma}
\begin{proof}
Let $H$ be the assumed Steiner tree; without loss of generality, $H$ is inclusion-wise minimal. Let $H_F$ denote the restriction of $H$ to the flower gadgets, and $H_{\fab}$ the restriction of $H$ to $\fab$. In general, $H_F$ is a forest. We observe that
$$M_7 > 3k^2 M_6 + 2k^2 \cdot (L M_5 + L(N-1)M_4 + N M_3 + M_2)$$
by the choice of $M$ and $N$. Since every edge in the flower gadgets has weight at least $M_7$ and is a multiple of $M_7$, it follows that $H_F$ has weight at most $k(k-1) \cdot (2t-4) \cdot tM_7$. As any path from a terminal in a flower gadget to $r$ contains a portal of that flower gadget, the minimality of $H$ implies that all trees of $H_F$ contain a portal of the corresponding gadget. In particular, for each $a \in [k-1],b\in[k]$, $H_{\fab}$ contains a portal of $\fab$. Then, by Theorem~\ref{thm:flower}, it follows that $H_{\fab}$ has weight at least $(2t-4) \cdot tM_7$. Since there are $k(k-1)$ flower gadgets, this implies that $H_F$ has weight at least $k(k-1) \cdot (2t-4) \cdot tM_7$, and thus weight exactly $k(k-1) \cdot (2t-4) \cdot tM_7$. Hence, for each $a \in [k-1], b \in [k]$, $H_{\fab}$ has weight exactly $(2t-4) \cdot tM_7$. Then Theorem~\ref{thm:flower} implies that each tree of $H_{\fab}$ (and in $H_F$) contains exactly one portal, and for any sequence of $t/2$ consecutive portals, at least one tree of $H_{\fab}$ has its portal there. Moreover, $H$ has weight at most $3k^2 M_6 + 2k^2 \cdot (L M_5 + L(N-1)M_4 + N M_3 + M_2)$ outside the flower gadgets.

Since each tree in $H_F$ contains exactly one portal, the path $P_a$ in $H$ between $r$ and $h_a$ for $a \in [k]$ cannot cross a flower gadget and is fully contained in $\{r,h_a\} \cup \left(\bigcup_{b \in [k]} \Gab\right) \cup \left(\bigcup_{b \in [k-1], i \in [n]} \fabi{i} \right)$. In particular, the path contains at least $k-1$ fuse vertices and two edges of weight $M_6$ incident on each of them, at least $k$ join edges of weight $M_6$, an edge of weight $M_6$ incident on $r$, an edge of weight $M_6$ incident on $h_a$. Hence, the path has weight at least $3kM_6$. Moreover, since $P_a$ cannot cross the flower gadgets, the $k$ paths $P_1,\ldots,P_k$ from $r$ to $h_1,\ldots,h_k$ respectively are internally vertex disjoint. Hence, the total weight of the aforementioned edges across all of the $k$ paths is $3k^2 M_6$. Since $M_6 > 2k^2 \cdot (L M_5 + L(N-1)M_4 + N M_3 + M_2)$ by the choice of $M$, it follows that $H$ contains no further edges of weight $M_6$. In particular, $H$ has weight at most $2k^2 \cdot (L M_5 + L(N-1)M_4 + N M_3 + M_2)$ in total in the gadgets $\gwest$ and $\geast$ for $a,b\in [k]$.

The preceding implies that $H$ contains exactly $k(k-1)$ fuse vertices, one for each $a \in [k], b \in [k-1]$, denoted $\fabi{\mi}$ for suitable $\mi \in [n]$. 
For notational convenience, define $\mi[a,0] = \mi[a,1]$ and $\mi[a,k] = \mi[a,k-1]$ for each $a \in [k]$.
Moreover, each of these fuse vertices has degree exactly~$2$ in $H$. Similarly, $h_1,\ldots,h_k$ each have degree exactly~$1$ in $H$, and $r$ has degree exactly~$k$. Finally, $H$ contains exactly one join edge $\ejj{\mj}$ for each $a,b \in [k]$ for certain $\mj \in [N]$. From this, we conclude that for each $a,b \in [k]$, $H$ contains exactly one $p$-portal and exactly one $q$-portal of each of $\gwest$ and $\geast$, specifically $\pwest{\mpwest}, \qwest{\mqwest}, \peast{\mpeast}, \qeast{\mqeast}$ for suitable $\mpwest, \mqwest, \mpeast, \mqeast \in [N]$.

Let $b \in [k]$ and $a \in [k-1]$. Since for any sequence of $t/2$ consecutive portals of $\fab$, at least one tree of $H_{\fab}$ has its portal there, it follows that one of $\wwest{l_{a,b}} = \lb l_{a,b},\,l_{a,b}+t/2-1\rb_t$ and one of $\weast[a+1,b]{l'_{a,b}} = \lb l'_{a,b}+t/2,\,l'_{a,b}-1\rb_t$ is in $H_{\fab}$ for certain $l_{a,b} \in [L]$ and $l'_{a,b} \in [L]$. From this and the placement of the dummy terminals, we conclude that for each $a,b \in [k]$, $H$ contains at least one $w$-portal of each of $\gwest$ and $\geast$, specifically $\wwest{\mwwest}$ and $\weast{\mweast}$ for suitable $\mwwest, \mweast \in [L]$.

Now note that for each $a,b \in [k]$, $H$ contains exactly one $p$-portal, exactly one $q$-portal, and at least one $w$-portal of each of $\gwest$ and $\geast$. Moreover, $H$ restricted to $\Gab$ (denoted $H_{\Gab}$) must be connected in order for a path from $r$ to $h_a$ to exist in $H$. Since only one join edge of $\Gab$ is in $H$, as established previously, it follows that $H$ restricted to $\gwest$ (denoted $H_{\gwest}$) and to $\geast$ (denoted $H_{\geast}$) must each be connected. Then Lemma~\ref{lem:g:veril}(\ref{lem:g:veril:ii}) implies that $H_{\gwest}$ and $H_{\geast}$ each have weight at least $L M_5 + L(N-1)M_4 + N M_3 + M_2$. Since $H$ has weight at most $2k^2 \cdot (L M_5 + L(N-1)M_4 + N M_3 + M_2)$ in total in the gadgets $\gwest$ and $\geast$ for $a,b\in [k]$, it follows that $H_{\gwest}$ and $H_{\geast}$ each have weight exactly $L M_5 + L(N-1)M_4 + N M_3 + M_2$ for each $a,b \in [k]$.

Let $a,b \in [k]$. Since $H_{\gwest}$ has weight exactly $L M_5 + L(N-1)M_4 + N M_3 + M_2$, is connected, and contains $\pwest{\mpwest}$, $\qwest{\mqwest}$, and $\wwest{\mwwest}$, it follows from Lemma~\ref{lem:g:veril}(\ref{lem:g:veril:ii}) that $\mpwest = \mqwest$ and that $H_{\gwest}$ contains only one selector edge, namely the $\mpwest$-selector incident on $\wwest{\mwwest}$. Consequently, $\wwest{\mwwest}$ is the only portal among $\wwest{1},\ldots,\wwest{L}$ that is in $H_{\gwest}$. Similar statements hold mutatis mutandis with respect to $H_{\geast}$. Since $H_{\Gab}$ contains exactly one join edge, it follows that $\mpwest = \mqwest = N-\mpeast+1 = N-\mqeast+1$. The construction of $\swest$ and $\seast$ implies that $\mwwest = \mweast$. Also, note that $H$ must contain the edge between $\fabi[a,b-1]{\mi[a,b-1]}$ and $\pwest{\mpwest}$ as well as the edge between $\fabi{\mi}$ and $\peast{\mpeast}$. The fact that $\mpwest = N-\mpeast+1$ implies that $\mi[a,b-1] = \mi$ by the definition of the fuse vertices. Hence, for each $a \in [k]$, it follows that $\mi[a,0] = \cdots = \mi[a,k]$. Set $x_a = \mi[a,0]$ for each $a \in [k]$.

Let $a \in [k-1], b \in [k]$. By the preceding paragraph, $\wwest{\mwwest}$ is the only portal among $\wwest{1},\ldots,\wwest{L}$ that is an element of $H_{\gwest}$, and the portal $\weast[a+1,b]{\mweast[a+1,b]}$ is the only one among $\weast[a+1,b]{1},\ldots,\weast[a+1,b]{L}$ that is in $H_{\geast[a+1,b]}$. This implies that $H_{\fab}$ has exactly two components. We previously established that $H_{\fab}$ has weight exactly $(2t-4) \cdot tM_7$. Then Theorem~\ref{thm:flower} implies that $H_{\fab}$ is canonical, meaning that the two portals of $H_{\fab}$ in $\fab$ are opposite. By the construction of $G_{\mathcal{M}}$, this implies that $\mwwest = \mweast[a+1,b]$. Recall that $\mwwest = \mweast$ and $\mwwest[a+1,b] = \mweast[a+1,b]$ was established in the previous paragraph. Hence, for each $b \in [k]$, it follows that $\mwwest[1,b] = \cdots = \mwwest[k,b] = \mweast[1,b] = \cdots = \mweast[k,b]$. Set $y_b = \mwwest[1,b]$ for each $b \in [k]$.

We claim that $x_1,\ldots,x_k,y_1,\ldots,y_k$ is a solution to the {\pGT} instance. Let $a,b\in [k]$. By definition, $\mwwest = \mweast = y_b$ and $\mi[a,b-1] = \mi = x_a$. Note that $\mpwest \in \lb (x_a-1)n+1, x_an \rb$ by the construction of $G_{\mathcal{M}}$ and by the fact that $\pwest{\mpwest}$ is the only $p$-portal of $\gwest$ in $H$. Since the $\mpwest$-selector incident on $\wwest{\mwwest} = \wwest{y_b}$ is in $H$, $\mpwest = (x_a-1)n+y_b$. The construction of $\gLVGS$ implies that $\mpwest \in \sswest{y_b}$. Then the construction of $\swest$, and specifically of $\sswest{y_b}$ implies that $(x_a,y_b) \in M_{a,b}$. The claim follows, and thus so does the lemma.
\end{proof}
The construction and the above lemmas immediately imply Theorem~\ref{thm:g:weighted}. The proof of Theorem~\ref{thm:g:bound} then quickly follows.

\begin{proof}[Proof of Theorem~\ref{thm:g:bound}]
Let $G_{\mathcal{M}}$ be the edge-weighted planar graph resulting from Theorem~\ref{thm:g:weighted}, with terminal set $T_{\mathcal{M}}$. Subdivide an edge $e$ of $G_{\mathcal{M}}$ of weight $w > 1$ exactly $w-1$ times, such that $e$ is replaced by a path of $w$ unit-weight edges. Call the resulting graph $\mathcal{G}_{\mathcal{M}}$.

The bound on the size is immediate from the fact that $G_{\mathcal{M}}$ has $O(k^2n^5)$ edges of weight $O(k^{14}n^{22})$ each. 

Note that there is a bijection between the faces of $\mathcal{G}_{\mathcal{M}}$ and of $G_{\mathcal{M}}$. Moreover, $T_{\mathcal{M}}$ is still present in $\mathcal{G}_{\mathcal{M}}$. Hence, the terminals of $T_{\mathcal{M}}$ can be covered by $k(k-1)+1$ faces of $\mathcal{G}_{\mathcal{M}}$.

The final property follows immediately from the subdivision of the edges in correspondence to their weights and from the corresponding property in Theorem~\ref{thm:g:weighted}.
\end{proof}

%\subfile{flowergadget.tex}

\section{Concluding Remarks}\label{sec:conc}
In this paper we gave an $2^{O(k)}n^{O(\sqrt{k})}$ time algorithm for \pPST, if the terminals are covered by $k$ faces, and showed this is almost optimal assuming the Exponential Time Hypothesis. The crucial idea in the algorithm was to study seperators in a graph with artificially added edges that enforce how connected components are divided. The crucial idea in the lower bound is the flower gadget that is a graph with all terminals on one face where an optimal forest consisting of two trees can divide the terminal set arbitrarily in two parts.

Several exciting questions remain. First, an interesting question is whether our techniques could inspire further progress in any of the studies that invoked the original algorithm of Erickson~\cite{DBLP:journals/mor/EricksonMV87}. For example in the mentioned approximation and kernelization algorithms~\cite{DBLP:conf/soda/BorradaileKK07,DBLP:conf/focs/PilipczukPSL14} the authors reduce the general {\pPST} to the case where terminals lie on one face. A natural direction to explore is to reduce the number of faces with terminals to more than one, and subsequently use the insights from this paper to aim for improved algorithms. It would also be interesting to see whether our techniques have consequences in the more geometric setting outlined by Provan~\cite{DBLP:journals/siamcomp/Provan88, DBLP:journals/networks/Provan88}.

Second, a natural question is whether the $2^{O(k)}$ term in our running time can be removed. This would significantly generalize the $n^{O(|\sqrt{T}|)}W$-time algorithm of~\cite{MarxPP17}.
A natural approach would be to combine our technique with the technique of~\cite{MarxPP17}, but it seems highly unclear in which graph one should consider separators.

\section*{Acknowledgments} We thank D\'aniel Marx, Marcin Pilipczuk and Micha\l{} Pilipczuk for allowing us to use their figures.

\bibliographystyle{abbrv}
\bibliography{refs}

\begin{thebibliography}{10}

\bibitem{DBLP:conf/soda/BateniCEHKM11}
M.~Bateni, C.~Chekuri, A.~Ene, M.~Hajiaghayi, N.~Korula, and D.~Marx.
\newblock Prize-collecting {S}teiner problems on planar graphs.
\newblock In {\em Proceedings of the Twenty-Second Annual {ACM-SIAM} Symposium
  on Discrete Algorithms, {SODA} 2011}, pages 1028--1049, 2011.

\bibitem{Bateni:2011:ASS:2027216.2027219}
M.~Bateni, M.~Hajiaghayi, and D.~Marx.
\newblock Approximation schemes for {S}teiner forest on planar graphs and
  graphs of bounded treewidth.
\newblock {\em J. ACM}, 58(5):21:1--21:37, 2011.

\bibitem{thesisbern}
M.~W. Bern.
\newblock {\em Network design problems: {S}teiner trees and spanning
  $k$-trees}.
\newblock PhD thesis, University of California, Berkeley, 1987.

\bibitem{DBLP:journals/anor/BernB91}
M.~W. Bern and D.~Bienstock.
\newblock Polynomially solvable special cases of the {S}teiner problem in
  planar networks.
\newblock {\em Annals {OR}}, 33(6):403--418, 1991.

\bibitem{DBLP:journals/siamcomp/BienstockM88}
D.~Bienstock and C.~L. Monma.
\newblock On the complexity of covering vertices by faces in a planar graph.
\newblock {\em {SIAM} J. Comput.}, 17(1):53--76, 1988.

\bibitem{Bjorklund:2007:FMM:1250790.1250801}
A.~Bj\"{o}rklund, T.~Husfeldt, P.~Kaski, and M.~Koivisto.
\newblock Fourier meets {M}\"{o}bius: Fast subset convolution.
\newblock In {\em Proceedings of the Thirty-ninth Annual ACM Symposium on
  Theory of Computing}, STOC '07, pages 67--74, New York, NY, USA, 2007. ACM.

\bibitem{DBLP:conf/soda/BorradaileKK07}
G.~Borradaile, C.~Kenyon{-}Mathieu, and P.~N. Klein.
\newblock A polynomial-time approximation scheme for {S}teiner tree in planar
  graphs.
\newblock In {\em Proceedings of the Eighteenth Annual {ACM-SIAM} Symposium on
  Discrete Algorithms, {SODA} 2007}, pages 1285--1294, 2007.

\bibitem{DBLP:conf/wads/BorradaileKM07}
G.~Borradaile, P.~N. Klein, and C.~Mathieu.
\newblock Steiner tree in planar graphs: An \emph{O}{(}\emph{n}log\emph{n}{)}
  approximation scheme with singly-exponential dependence on epsilon.
\newblock In {\em Algorithms and Data Structures, 10th International Workshop,
  {WADS} 2007, Proceedings}, pages 275--286, 2007.

\bibitem{BOUCHITTE200134}
V.~Bouchitt\'e, F.~Mazoit, and I.~Todinca.
\newblock Treewidth of planar graphs: connections with duality.
\newblock {\em Electronic Notes in Discrete Mathematics}, 10:34 -- 38, 2001.
\newblock Comb01, Euroconference on Combinatorics, Graph Theory and
  Applications.

\bibitem{DBLP:journals/algorithmica/ChenW03}
D.~Z. Chen and X.~Wu.
\newblock Efficient algorithms for k-terminal cuts on planar graphs.
\newblock {\em Algorithmica}, 38(2):299--316, 2004.

\bibitem{Chen00}
D.~Z. Chen and J.~Xu.
\newblock Shortest path queries in planar graphs.
\newblock In {\em Proceedings of the Thirty-second Annual {ACM} Symposium on
  Theory of Computing}, STOC '00, pages 469--478. ACM, 2000.

\bibitem{Cheng2000}
S.-W. Cheng.
\newblock {\em Exact {S}teiner Trees in Graphs and Grid Graphs}, pages
  137--162.
\newblock Springer US, Boston, MA, 2000.

\bibitem{DBLP:books/sp/CyganFKLMPPS15}
M.~Cygan, F.~V. Fomin, L.~Kowalik, D.~Lokshtanov, D.~Marx, M.~Pilipczuk,
  M.~Pilipczuk, and S.~Saurabh.
\newblock {\em Parameterized Algorithms}.
\newblock Springer, 2015.

\bibitem{Demaine:2005:FAK:1077464.1077468}
E.~D. Demaine, F.~V. Fomin, M.~Hajiaghayi, and D.~M. Thilikos.
\newblock Fixed-parameter algorithms for (k, r)-center in planar graphs and map
  graphs.
\newblock {\em ACM Trans. Algorithms}, 1(1):33--47, July 2005.

\bibitem{DBLP:journals/jacm/DemaineFHT05}
E.~D. Demaine, F.~V. Fomin, M.~Hajiaghayi, and D.~M. Thilikos.
\newblock Subexponential parameterized algorithms on bounded-genus graphs and
  \emph{H}-minor-free graphs.
\newblock {\em J. {ACM}}, 52(6):866--893, 2005.

\bibitem{DBLP:books/daglib/0030488}
R.~Diestel.
\newblock {\em Graph Theory, 4th Edition}, volume 173 of {\em Graduate texts in
  mathematics}.
\newblock Springer, 2012.

\bibitem{DBLP:journals/networks/DreyfusW71}
S.~E. Dreyfus and R.~A. Wagner.
\newblock The {S}teiner problem in graphs.
\newblock {\em Networks}, 1(3):195--207, 1971.

\bibitem{Du:2008:STP:1628718}
D.~Du and X.~Hu.
\newblock {\em Steiner Tree Problems In Computer Communication Networks}.
\newblock World Scientific Publishing Co., Inc., River Edge, NJ, USA, 2008.

\bibitem{DBLP:journals/mor/EricksonMV87}
R.~E. Erickson, C.~L. Monma, and A.~F. Veinott~Jr.
\newblock Send-and-split method for minimum-concave-cost network flows.
\newblock {\em Math. Oper. Res.}, 12(4):634--664, 1987.

\bibitem{DBLP:journals/algorithmica/FominGKLS13}
F.~V. Fomin, F.~Grandoni, D.~Kratsch, D.~Lokshtanov, and S.~Saurabh.
\newblock Computing optimal {S}teiner trees in polynomial space.
\newblock {\em Algorithmica}, 65(3):584--604, 2013.

\bibitem{DBLP:conf/icalp/FominKLPS15}
F.~V. Fomin, P.~Kaski, D.~Lokshtanov, F.~Panolan, and S.~Saurabh.
\newblock Parameterized single-exponential time polynomial space algorithm for
  {S}teiner tree.
\newblock In {\em Automata, Languages, and Programming - 42nd International
  Colloquium, {ICALP} 2015, Proceedings, Part {I}}, pages 494--505, 2015.

\bibitem{DBLP:conf/focs/FominLMPPS16}
F.~V. Fomin, D.~Lokshtanov, D.~Marx, M.~Pilipczuk, M.~Pilipczuk, and
  S.~Saurabh.
\newblock Subexponential parameterized algorithms for planar and
  apex-minor-free graphs via low treewidth pattern covering.
\newblock In {\em {IEEE} 57th Annual Symposium on Foundations of Computer
  Science, {FOCS} 2016}, pages 515--524, 2016.

\bibitem{Ford-Fulkerson}
L.~R. Ford and D.~R. Fulkerson.
\newblock {Maximal Flow through a Network.}
\newblock {\em Canadian Journal of Mathematics}, 8:399--404, 1954.

\bibitem{Frederickson91}
G.~N. Frederickson.
\newblock Planar graph decomposition and all pairs shortest paths.
\newblock {\em J. ACM}, 38(1):162--204, 1991.

\bibitem{DBLP:journals/mst/FuchsKMRRW07}
B.~Fuchs, W.~Kern, D.~M{\"{o}}lle, S.~Richter, P.~Rossmanith, and X.~Wang.
\newblock Dynamic programming for minimum {S}teiner trees.
\newblock {\em Theory Comput. Syst.}, 41(3):493--500, 2007.

\bibitem{annalsofDMsteinertree}
F.~K. Hwang, D.~S. Richards, and P.~Winter.
\newblock {\em The {S}teiner Tree Problem}, volume~53 of {\em Annals of
  Discrete Mathematics}.
\newblock Elsevier, 1992.

\bibitem{KrauthgamerLR19}
R.~Krauthgamer, J.~R. Lee, and H.~I. Rika.
\newblock Flow-cut gaps and face covers in planar graphs.
\newblock In {\em Proceedings of the Thirtieth Annual {ACM-SIAM} Symposium on
  Discrete Algorithms {SODA} 2019}, 2019.
\newblock To appear.

\bibitem{DBLP:journals/corr/KrauthgamerR17}
R.~Krauthgamer and I.~Rika.
\newblock Refined vertex sparsifiers of planar graphs.
\newblock {\em CoRR}, abs/1702.05951, 2017.

\bibitem{levin}
A.~Y. Levin.
\newblock Algorithm for the shortest connection of a group of graph vertices.
\newblock {\em Soviet Mathematics Doklady}, 12:1477--1481, 1971.

\bibitem{DBLP:conf/stoc/LokshtanovN10}
D.~Lokshtanov and J.~Nederlof.
\newblock Saving space by algebraization.
\newblock In {\em Proceedings of the 42nd {ACM} Symposium on Theory of
  Computing, {STOC} 2010}, pages 321--330, 2010.

\bibitem{doi:10.1002/net.3230200110}
B.~Marshall.
\newblock Faster exact algorithms for {S}teiner trees in planar networks.
\newblock {\em Networks}, 20(1):109--120, 1990.

\bibitem{DBLP:conf/icalp/Marx12}
D.~Marx.
\newblock A tight lower bound for planar multiway cut with fixed number of
  terminals.
\newblock In A.~Czumaj, K.~Mehlhorn, A.~M. Pitts, and R.~Wattenhofer, editors,
  {\em Automata, Languages, and Programming - 39th International Colloquium,
  {ICALP} 2012, Proceedings, Part {I}}, volume 7391 of {\em Lecture Notes in
  Computer Science}, pages 677--688. Springer, 2012.

\bibitem{DBLP:conf/esa/MarxP15}
D.~Marx and M.~Pilipczuk.
\newblock Optimal parameterized algorithms for planar facility location
  problems using {V}oronoi diagrams.
\newblock In {\em Algorithms - {ESA} 2015 - 23rd Annual European Symposium,
  Proceedings}, pages 865--877, 2015.

\bibitem{MarxPP17}
D.~Marx, M.~Pilipczuk, and M.~Pilipczuk.
\newblock On subexponential parameterized algorithms for {S}teiner tree and
  directed subset {TSP} on planar graphs.
\newblock In {\em Proceedings of the 59th Annual IEEE Symposium on Foundations
  of Computer Science}, 2018.
\newblock To appear. See also~\url{https://arxiv.org/abs/1707.02190}.

\bibitem{DBLP:journals/siamcomp/MatsumotoNS85}
K.~Matsumoto, T.~Nishizeki, and N.~Saito.
\newblock An efficient algorithm for finding multicommodity flows in planar
  networks.
\newblock {\em {SIAM} J. Comput.}, 14(2):289--302, 1985.

\bibitem{DBLP:journals/algorithmica/Nederlof13}
J.~Nederlof.
\newblock Fast polynomial-space algorithms using inclusion-exclusion.
\newblock {\em Algorithmica}, 65(4):868--884, 2013.

\bibitem{OkamuraS1981}
H.~Okamura and P.~D. Seymour.
\newblock Multicommodity flows in planar graphs.
\newblock {\em Journal of Combinatorial Theory, Series B}, 31:75--81, 1981.

\bibitem{DBLP:conf/stacs/PilipczukPSL13}
M.~Pilipczuk, M.~Pilipczuk, P.~Sankowski, and E.~J. van Leeuwen.
\newblock Subexponential-time parameterized algorithm for {S}teiner tree on
  planar graphs.
\newblock In {\em 30th International Symposium on Theoretical Aspects of
  Computer Science, {STACS} 2013}, pages 353--364, 2013.

\bibitem{DBLP:conf/focs/PilipczukPSL14}
M.~Pilipczuk, M.~Pilipczuk, P.~Sankowski, and E.~J. van Leeuwen.
\newblock Network sparsification for {S}teiner problems on planar and
  bounded-genus graphs.
\newblock In {\em 55th {IEEE} Annual Symposium on Foundations of Computer
  Science, {FOCS} 2014}, pages 276--285, 2014.

\bibitem{promel2012steiner}
H.~Pr{\"o}mel and A.~Steger.
\newblock {\em The {S}teiner Tree Problem: A Tour through Graphs, Algorithms,
  and Complexity}.
\newblock Advanced Lectures in Mathematics. Vieweg+Teubner Verlag, 2012.

\bibitem{DBLP:journals/siamcomp/Provan88}
J.~S. Provan.
\newblock An approximation scheme for finding {S}teiner trees with obstacles.
\newblock {\em {SIAM} J. Comput.}, 17(5):920--934, 1988.

\bibitem{DBLP:journals/networks/Provan88}
J.~S. Provan.
\newblock Convexity and the {S}teiner tree problem.
\newblock {\em Networks}, 18(1):55--72, 1988.

\end{thebibliography}

\end{document}